\documentclass[a4paper,11pt]{article}

\usepackage[german, english]{babel}
\usepackage[utf8]{inputenc}

\usepackage{amssymb}
\usepackage{amsmath}
\usepackage{amsthm}
\usepackage{amsfonts}
\usepackage{mathtools}
\usepackage{thmtools}
\usepackage{latexsym}
\usepackage{graphicx}
\usepackage{bm}  
\usepackage{overpic}
\usepackage[normalem]{ulem} 
\usepackage[colorlinks=false,linkbordercolor={0 0 1}]{hyperref} 
\usepackage{exscale} 
\usepackage[usenames,dvipsnames]{color} 
\usepackage{enumerate}
\usepackage{xfrac} 
\usepackage{booktabs}
\usepackage{tikz}
\usetikzlibrary{matrix, fit, trees, positioning}
\usepackage{layout}
\usepackage{geometry} 
\usepackage{subcaption}
\usepackage{adjustbox}
\usepackage{framed}

\usepackage{natbib}
\AtBeginDocument{\renewcommand{\harvardand}{and}}
\renewcommand\harvardurl{URL: \url}  
\newcommand*{\doi}[1]{\href{https://doi.org/#1}{https://doi.org/#1}}

\def\wt{\widetilde}

\def\eps{\varepsilon}

\def\ignore#1{}

\def\cB{\mathcal B}
\def\cC{\mathcal C}

\def\cI{\mathcal I}

\def\cM{\mathcal M}
\def\cP{\mathcal P}

\def\cS{\mathcal S}

\def\bN{\mathbb N}

\def\bR{\mathbb R}
\def\bT{\mathbb T}
\def\bZ{\mathbb Z}

\DeclareMathOperator*{\argmax}{arg\,max}

\newcommand{\interior}[1]{%
	{\kern0pt#1}^{\mathrm{o}}%
}
\newcommand{\1}{{\rm 1\hspace*{-0.4ex}
		\rule{0.1ex}{1.52ex}\hspace*{0.2ex}}}
\def\Ind#1{\1_{_{\scriptstyle #1}}}			

\DeclarePairedDelimiter\abs{\lvert}{\rvert}   

\DeclarePairedDelimiter\floor{\lfloor}{\rfloor}

\theoremstyle{plain} 		
\newtheorem{theorem}{Theorem}[section]

\newtheorem{lemma}[theorem]{Lemma}

\newtheorem{definition}[theorem]{Definition}

\theoremstyle{definition}	
\newtheorem{example}[theorem]{Example}

\newtheorem{remark}[theorem]{Remark}
\newtheorem{assumption}[theorem]{Assumption}

\numberwithin{equation}{section}
\numberwithin{figure}{section}

\renewcommand{\baselinestretch}{1}\normalsize

\newcommand\extrafootertext[1]{%
	\bgroup
	\renewcommand\thefootnote{\fnsymbol{footnote}}%
	\renewcommand\thempfootnote{\fnsymbol{mpfootnote}}%
	\footnotetext[0]{#1}%
	\egroup
}

\begin{document}
	
\title{Insider trading in discrete time Kyle games}
\author{Christoph K\"uhn$^\star$ \and Christopher Lorenz$^\star$}
\date{%
	$^\star$Institute of Mathematics,
	Goethe University Frankfurt, \\
	D-60054 Frankfurt am Main, Germany \\
	\texttt{\{ckuehn,lorenz\}@math.uni-frankfurt.de}\\[2ex]%
}

\maketitle

\begin{abstract}
We present a new discrete time version of 
\citeauthor{kyle_continuous_1985}'s \citeyearpar{kyle_continuous_1985} classic model of insider trading, formulated as a generalised extensive form game.
	The model has three kinds of traders: an insider, random noise traders, and a market maker. The insider aims to exploit her informational advantage and maximise expected profits while the market maker observes the total order flow and sets prices accordingly.
	
	First, we show how the multi-period model with finitely many pure strategies can be reduced to a (static) social system in the sense of 
	\cite{debreu_social_1952} and prove the existence of a sequential Kyle equilibrium, following \cite{kreps_sequential_1982}. 
This works for any probability distribution with finite support of the noise trader's demand and the true value, and for any finite 
information flow of  the insider.	
In contrast to \cite{kyle_continuous_1985} with normal distributions, 
equilibria exist in general only in mixed strategies and not in pure strategies.

In the single-period model we establish bounds for the insider's strategy in equilibrium.
	Finally, we prove the existence of an equilibrium for the game with a continuum of actions, by considering an approximating sequence of games with finitely many actions. Because of the lack of compactness of the set of measurable price functions, standard infinite-dimensional fixed point theorems are not applicable.
\end{abstract}


\extrafootertext{\emph{Acknowledgments.} We would like to thank Umut \c{C}etin, Ralph Neininger, Emanuel Rapsch, Marius Schmidt, and participants at the Workshop on Stochastic Dynamic Games, University of Kiel (2023), for useful comments and two anonymous referees for valuable suggestions that led to an improvement of the presentation.}


\section{Introduction}

In this paper, we study the discrete time version of \citeauthor{kyle_continuous_1985}'s \citeyearpar{kyle_continuous_1985} classic model of a specialist market with asymmetrically informed agents.
We propose to model it as a generalised extensive form game that is explained below.
The Kyle model features three (kinds of) traders: an informed trader, called the \textit{insider}, random noise traders, and a market maker. 

While the insider acts strategically in order to exploit her informational advantage about the true value of an asset, the uninformed noise trader submits orders stochastically independent. The market maker cannot distinguish between informed and uninformed trades, and thus he sets prices depending only on the total order flow, from which he learns conditional distributions of the true value.
In a \emph{Kyle equilibrium}, the insider trades optimally given the price function of the market maker, and given the insider's strategy the market maker's prices are rational (in the case of risk-neutrality, they fulfill a zero-profit condition). 

Kyle's model has been a workhorse model for understanding the role of asymmetric information in the formation of prices. For practitioners it can serve as a conceptual tool to explain the incentives of market makers to provide different prices to different market participants based on their level of informedness. This can be one reason why brokers offer better terms to retail customers, a phenomenon widely discussed in connection with 
\emph{payment for order flow} (we refer to \cite{cetin_order_2022} and for a recent empirical study see \cite{lynch_price_2022}). 

In his seminal paper, \cite{kyle_continuous_1985} proposes a discrete time model assuming normal distributions and a risk-neutral insider and market maker. The continuous time model was first comprehensively studied in \cite{back_insider_1992} under the distributional assumption that noise trades follow a Brownian motion. The Kyle-Back model has been extended in various directions, and we can only cite a selection here. 

A general principle to establish equilibria follows the \emph{inconspicuous trade} ansatz. 
The resulting equilibrium insider strategy is inconspicuous in the sense that the accumulated total demand process has the same law as the accumulated demand of the noise trader alone. 
Furthermore, the true value, which is independent of the noise
trader's demand, has to be a nondecreasing function of the accumulated total
demand at maturity. 
The market maker's price function of the total demand and time is chosen such that {\em any} insider demand process that leads to the terminal accumulated total demand described above is optimal.
In \cite{cetin_point_2013} the existence of an inconspicuous equilibrium is shown if noise trades follow a Poisson process instead of Brownian motion. Further extensions include trading with default risk in \cite{campi_insider_2007} and a random trading horizon in \cite{cetin_financial_2018}. 
The case of dynamic private information where the insider's information evolves over time is considered in  \cite{back_long-lived_1998} and \cite{danilova_stock_2010}.
For a thorough overview we refer to the monograph \cite{cetin_dynamic_2018}.

The inconspicuous trade ansatz does not work for risk-averse market makers. In the continuous time model, \cite{cetin_markovian_2016} prove the existence of a non-inconspicuous equilibrium that is derived from a fixed point (see the discussion below). In \cite{back_optimal_2021} there are several assets, and the distribution of the true value and the accumulated total demand at maturity are coupled through an optimal transport map.

In the single-period model, under normality assumptions, \cite{kyle_continuous_1985} shows that there exists a unique equilibrium in the class of affine-linear equilibria (if the insider's demand in equilibrium is an affine function in the true value, then the price function is automatically affine in the total demand). \cite{mclennan_uniqueness_2017} extend uniqueness to a broader class of equilibria. \cite{cetin_is_2023} study whether Kyle's affine-linear equilibrium is stable for different trading times.
\cite{rochet_insider_1994} study a related single-period model where the insider can observe the trade of the noise trader and show the existence of a unique equilibrium while relaxing the assumption of normality. \cite{kramkov_optimal_2022} further relax distributional assumptions by connecting the problem to the dual problem of a certain class of optimal martingale transport problems.

\smallskip

In the present paper, we specify the discrete time Kyle model with a risk-neutral insider and market maker as a {\em generalised extensive form game}. The approach is purely discrete but has the advantage that the
proof of existence of an equilibrium is {\em not} based on the normal distribution or on any other special probability distribution as in most of the previous literature.
Rather, our approach works for {\em any} probability distribution with finite support of the noise trader's demand and the true value, and for {\em any} finite 
information flow of  the insider.
Such a flexible mathematical framework is particularly useful for realistic models with a market participant possessing superior information that is not perfect and evolves over time. 
It turns out that equilibria exist in general only in mixed strategies of the insider (see Example~\ref{example:no_pure_equilibrium}). This is in contrast to the classic Kyle model with normal distributions, in which equilibria exist in pure strategies.
In addition, we can characterise the set of all equilibria (see Remark~\ref{25.6.2024.3}), which enables a systematic study of
Kyle equilibria.
%
%

The game tree, the so-called extensive form, gets to the heart of the decision theoretical nature of the Kyle model: nodes in the tree signify the realised history of order and information flows, based on which the insider makes a trading decision. 
Even though the market maker ``does not explicitly maximize any particular objective'' (\cite{kyle_continuous_1985}), rational pricing can be considered as an action of the market maker that is restricted by the assumed behaviour strategy of the insider (we refer to Remark~\ref{remark:social_system} for a detailed discussion). 
Our purely discrete approach makes the assumptions of the Kyle model very explicit: The trades of the insider and the noise trader occur truly simultaneously, and the insider is not able to observe the current trade of the noise trader. After the orders are submitted the market maker sets the price. 

In continuous time models the chronological sequence is a bit less explicit. In most models, accumulated trades of the noise trader are given by a Brownian motion or another diffusion, and admissible insider strategies are of finite variation. In this constellation, trading volumes are of different orders, and  
the chronological sequence of trades becomes irrelevant in the continuous time limit.
On the other hand, there are continuous time generalizations of classic game trees, but with serious and in some cases even insoluble conceptual problems caused by the lack of immediate successors of nodes. We refer to \cite{alos-ferrer_theory_2016}, especially Subsection~5.7, for a detailed discussion. 
It is beyond the scope of the paper to relate this to continuous time Kyle models. In continuous time {\em timing games} (in which each player has only one move) \cite{riedel_subgame-perfect_2017} 
show the existence of a subgame-perfect Nash equilibrium. 

Our first main result (Theorem~\ref{theorem:sequential_Kyle_equilibrium}) shows the existence of a \emph{sequential Kyle equilibrium} in the sense of \cite{kreps_sequential_1982}. Since the market maker does not have perfect information,
a sequential equilibrium is based on his beliefs regarding the true value that generalise conditional probabilities derived from Bayes' rule.
The concept of a sequential equilibrium is a refinement of a subgame-perfect Nash equilibrium. Subgame-perfectness rules out non-credible strategic plans that are not realised in equilibrium (cf. Subsection~\ref{section:sequential_equilibrium}).   

The proof of Theorem~\ref{theorem:sequential_Kyle_equilibrium} relies on a fixed point theorem for a self-correspondence acting on the insider's dynamic strategies and the market maker's pricing functions. This is conceptually different to the fixed point theorem in the continuous time model of \cite{cetin_markovian_2016} in which the functional acts on the probability distributions of the accumulated total demand at maturity, and the dynamic quantities are constructed from a fixed point distribution.    

In Section~\ref{section:single_period_structure}, we establish basic properties of equilibria in single-period Kyle games.
The insider's demand is nondecreasing in the true value, but there need not exist a nondecreasing price function of the market maker (Example~\ref{example:S_not_monotone}).
Since the insider uses the noise trades as camouflage to remain unobserved by the market maker, 
the following obvious question arises: Are the insider's demands in equilibrium uniformly bounded for a family of models such that the
(exogenous) noise trader's demands lie in  a fixed bounded interval? The answer turns out to be positive if the probabilities of the insider's demand at the boundaries of the interval are bounded away from zero (see Theorem~\ref{theorem:single_period_structure} that also provides the range of insider's demands explicitly).  

In Section~\ref{section:continuous_state_game}, we prove the existence of an equilibrium in the continuous state game (Theorem~\ref{theorem:continuous_trade}) by considering an approximating sequence of discrete state games.  In view of Example~\ref{example:S_not_monotone}, admissible price functions only need to be measurable but not continuous in the total demand.
This leads to the problem that the set of price functions (equipped with pointwise convergence almost everywhere) is not sequentially compact, and standard infinite-dimensional fixed point theorems of Schauder-Tychonoff's or Kakutani-Fan's type (see, e.g., Theorem~10.1 and Theorem~13.1 in \cite{pata_fixed_2019}, respectively) cannot be applied directly to the continuous game. 

\section{The Kyle model as extensive form game} \label{section:extensive_form_game}

We specify the discrete time, discrete state Kyle model as a generalised extensive form game. 
Extensive form games have a reputation of being notationally burdensome, but possess great interpretive power that is rooted in its tree structure.
In general, we follow the notation of \cite[Chapter~3]{gonzalez-diaz_introductory_2010} in the spirit of \cite{selten_reexamination_1975}. In line with the literature, we label quantities related to the insider by $X$, related to the noise trader by $Z$, and related to the market maker by $Y$.

\subsection{Extensive form}

Trading takes place over the course of $T\in\bN$ trading rounds at times~$t= 1, 2, \ldots, T$. The insider possesses private information about the true value of the asset, not known to the noise trader or market maker. There are $N\in\bN$ fundamental information states~$1,\ldots,N$, and the true value is a mapping 
\begin{equation*}
	v : I \to \bR, 	\quad i \mapsto v^i,\qquad\mbox{where}\ I := \{1, 2, \ldots, N \},
\end{equation*}
such that $v^1 \ge v^2 \ge \ldots \ge v^N$. The flow of information is specified exogenously by a refining sequence of partitions of $I$:
\begin{equation*}
	\cI_t := \{ I_{t, 1}, \ldots, I_{t, N_t} \} \subseteq 2^I,  \qquad t \in \{1, 2, \ldots, T\},\ N_t \in \{1,2,\ldots,N\}.
\end{equation*}
Refining means that for every element~$I_{t, i} \in \cI_t$ with $t\ge 2$, there is $I_{t-1, j} \in \cI_{t-1}$, called a predecessor, such that $I_{t-1, j} \supseteq I_{t, i}$. A subset~$I_{t, i} \subseteq I$ represents the information states in $I$ which are still attainable at time $t$. There can be multiple states $i$ leading to the same true value. This allows to model the change of conditional probabilities of the true value over time in a flexible way. By contrast, for some injective function $i\mapsto v^i$, the conditional probability of a state can only increase or drop to zero over time. For notational convenience we define 
\begin{equation*}
	\cI_{T+1} := \big\{\{1\}, \{2\}, \ldots, \{N\}\big\}.
\end{equation*}
The full revealing of the true value is not a restriction of generality since at the future fictive date~$T+1$, trading has already concluded. Even the extreme case that the insider has no information at all until the end is not excluded: the sequence of partitions would be given by $\mathcal{I}_1 = \mathcal{I}_2 = \ldots = \mathcal{I}_T =\{I\}$. 

Each round of trading can be broken down into two steps. Firstly, new information about the true value of the asset is revealed to the insider. Secondly, both insider and noise trader \emph{simultaneously} trade a discrete quantity of shares from the set
\begin{equation*}
	E_X := \left\{x^1, \ldots, x^K\right\}\subseteq\bR \quad\mbox{and}\quad E_Z := \left\{z^1, \ldots, z^L\right\}\subseteq\bR,\quad \mbox{resp., where}\ K,L\in\bN.
\end{equation*}
The set~$E_X$ could be, e.g., multiples of the minimal (fractional) order size.
The market maker only observes the sum and sets a price. 
Since the noise trader does not trade strategically, her action can be placed after that of the insider in the tree in Definition~\ref{definition:game_tree} below (see also Figure \ref{figure:game_tree}).
This sequence of information disclosures, insider and noise trades spans the game tree, formally given by: 

\begin{definition}[Game tree] \label{definition:game_tree}
	The \emph{game tree}~$(\bT, E)$ is given by the set of nodes
	\begin{equation*}
		\begin{aligned}
			\bT &:= \{ r, I_{1, i_1}, (I_{1, i_1}, x_1), (I_{1, i_1}, x_1, z_1), (I_{1, i_1}, x_1, z_1, I_{2, i_2}), (I_{1,i_1}, x_1, z_1, \ldots, I_{t,i_t}, x_{t}),\\
			&\qquad (I_{1,i_1}, x_1, z_1, \ldots, I_{t,i_t}, x_{t}, z_{t}), (I_{1,i_1}, x_1, z_1, \ldots, I_{t,i_t}, x_{t}, z_{t}, I_{t+1, i_{t+1}}): \\
			&\qquad t \in \{2, \ldots, T\},\; (x_1, \ldots, x_T) \in (E_X)^T, \; (z_1, \ldots, z_T) \in (E_Z)^T, \\
			&\qquad (I_{1,i_1}, \ldots, I_{T+1,i}) \in \cI_1 \times \cdots \times \cI_{T+1}\; \text{ such that } I_{1,i_1} \supseteq \ldots \supseteq I_{T+1,i} \},
		\end{aligned}
	\end{equation*}
	where $r$ denotes the root node, and by the set of canonical edges~$E \subseteq \bT \times \bT$, meaning that, for example,
	$(r,I_{1, i_1})\in E$ and $(I_{1, i_1}, (I_{1, i_1}, x_1)) \in E$.
\end{definition}%
\noindent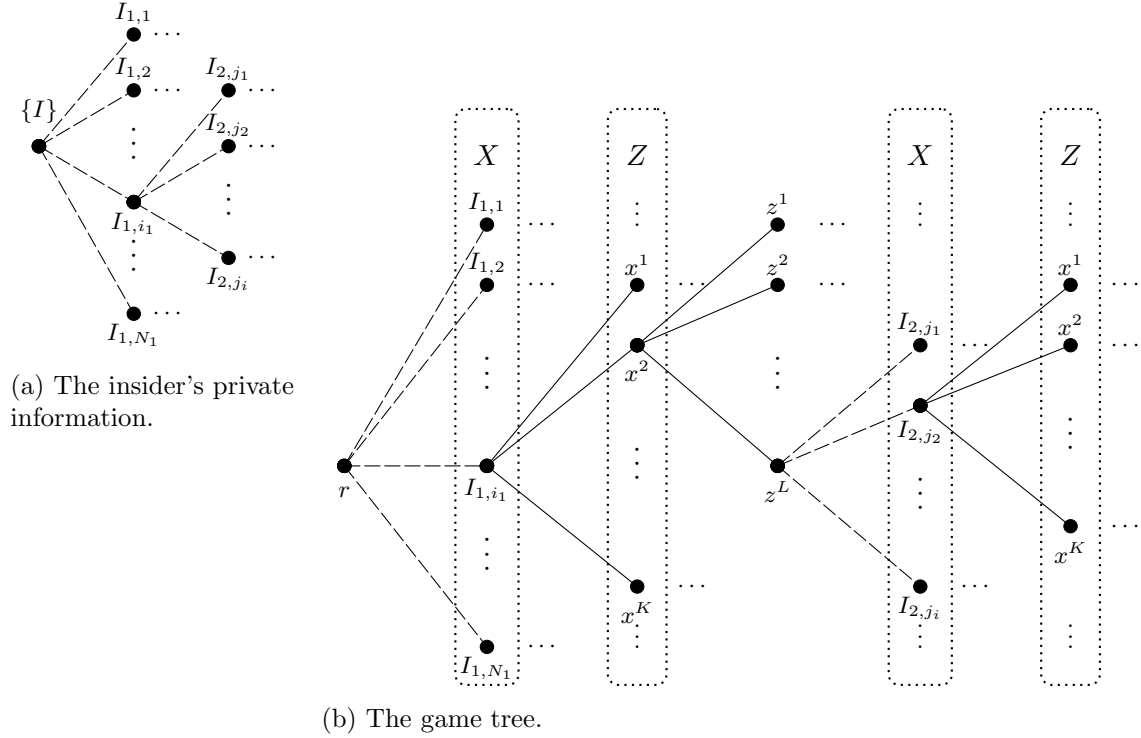
\begin{figure} 
	\tikzset{
		dot/.style={draw, solid, fill, circle, inner sep=0.0cm, outer sep=0cm, minimum size=5pt}
	}
	\tikzstyle{every node}=[font=\footnotesize]
	\tikzstyle{densely dashed}=[dash pattern=on 6pt off 1.5pt]
	\begin{subfigure}[c][][c]{0.25\textwidth}
		\begin{tikzpicture}[>=stealth]
			\matrix (m) [
			matrix of nodes,
			column sep=0.9cm,
			row sep=0.45cm
			]
			{
				& .		& 		\\ 
				& .   	& .		\\ 
				.	&   	& .		\\ 
				& .  	& 		\\ 
				&   	& .		\\ 
				& . 	& 		\\ 
			};
			\draw[densely dashed] (m-3-1.center) node [dot]{} node [above=5pt]{$\{I\}$} -- (m-1-2.center) node [dot]{} node [above=0pt] {$I_{1,1}$} node [right=-6pt] {$\quad\cdots$};		
			\draw[densely dashed] (m-3-1.center) node [dot]{} -- (m-2-2.center) node [dot]{} node [above=0pt] {$I_{1,2}$} node [right=-6pt] {$\quad\cdots$};
			\draw[densely dashed] (m-3-1.center) node [dot]{} -- (m-4-2.center) node [dot]{} node [below=0pt] {$I_{1,i_1}$};
			\draw[densely dashed] (m-3-1.center) node [dot]{} -- (m-6-2.center) node [dot]{} node [below=0pt] {$I_{1,N_1}$} node [right=-6pt] {$\quad\cdots$};
			\path (m-2-2) -- (m-4-2) node [font=\large, midway, sloped] {$\dots$};
			\path (m-4-2) -- (m-6-2) node [font=\large, midway, sloped] {$\dots$};
			\draw[densely dashed] (m-4-2.center) node [dot]{} -- (m-2-3.center) node [dot]{} node [midway,above,pos=.60] {} node [above=0pt] {$I_{2,j_1}$} node [right=-6pt] {$\quad\cdots$};
			\draw[densely dashed] (m-4-2.center) node [dot]{} -- (m-3-3.center) node [dot]{} node [midway,below,pos=.60] {} node [above=0pt] {$I_{2,j_2}$} node [right=-6pt] {$\quad\cdots$};
			\draw[densely dashed] (m-4-2.center) node [dot]{} -- (m-5-3.center) node [dot]{} node [midway,below,pos=.60] {} node [below=0pt] {$I_{2,j_i}$} node [right=-6pt] {$\quad\cdots$};
			\path (m-3-3) -- (m-5-3) node [font=\large, midway, sloped] {$\dots$};
		\end{tikzpicture}
		\caption{The insider's private information.}
		\vspace{-3cm}							
	\end{subfigure}	
	\hspace{0\textwidth} 
	\begin{subfigure}[t][][c]{.7\textwidth} 
		\begin{tikzpicture}[>=stealth]
			\matrix (m) [%
			matrix of nodes,
			column sep=1.4cm, 
			row sep={.80cm,between origins} 
			]
			{
				& {\normalsize $X$}	& {\normalsize $Z$}	& 		& {\normalsize $X$}	& {\normalsize $Z$}    \\ 
				& .			& \vdots	& .			& \vdots	& \vdots	\\ 
				& .   		& .			& .  		& 			& .			\\ 
				&   		& .			&       	& .			& .			\\ 
				&   		& 			&   		& .			&			\\ 
				. 	& .  		& 			& .	 		& 			&			\\ 
				&   		& 			&   		& 			& .			\\ 
				&   		& .			&   		& .			& 			\\ 
				& . 		& \vdots	&   		& \vdots	&	\vdots	\\ 
			};
			\draw[densely dashed] (m-6-1.center) node [dot]{} node [below=4pt]{$r$} -- (m-2-2.center) node [dot]{} node [above=0pt] {$I_{1,1}$} node [right=2pt] {$\quad\cdots$};
			\draw[densely dashed] (m-6-1.center) node [dot]{} -- (m-3-2.center) node [dot]{} node [above=0pt] {$I_{1,2}$} node [right=2pt] {$\quad\cdots$};
			\draw[densely dashed] (m-6-1.center) node [dot]{} -- (m-6-2.center) node [dot]{} node [below=1pt] {$I_{1,i_1}$};
			\draw[densely dashed] (m-6-1.center) node [dot]{} -- (m-9-2.center) node [dot]{} node [below=0pt] {$I_{1,N_1}$} node [right=2pt] {$\quad\cdots$};
			\path (m-3-2) -- (m-6-2) node [font=\large, midway, sloped] {$\dots$};
			\path (m-6-2) -- (m-9-2) node [font=\large, midway, sloped] {$\dots$};
			\draw[-] (m-6-2.center) node [dot]{} -- (m-3-3.center) node [dot]{} node [above=0pt] {$x^1$} node [right=2pt] {$\quad\cdots$}; 
			\draw[-] (m-6-2.center) node [dot]{} -- (m-4-3.center) node [dot]{} node [
			below=2pt] {$x^2$};
			\draw[-] (m-6-2.center) node [dot]{} -- (m-8-3.center) node [dot]{} node [below=2pt] {$x^K$} node [right=2pt] {$\quad\cdots$}; 
			\path (m-4-3) -- (m-8-3) node [font=\large, midway, sloped] {$\dots$};
			\draw[-] (m-4-3.center) node [dot]{} -- (m-2-4.center) node [dot]{} node [above=0pt] {$z^1$} node [right=2pt] {$\quad\cdots$};
			\draw[-] (m-4-3.center) node [dot]{} -- (m-3-4.center) node [dot]{} node [above=0pt] {$z^2$} node [right=2pt] {$\quad\cdots$}; 
			\draw[-] (m-4-3.center) node [dot]{} -- (m-6-4.center) node [dot]{} node [below=2pt] {$z^L$}; 
			\path (m-3-4) -- (m-6-4) node [font=\large, midway, sloped] {$\dots$};
			\draw[densely dashed] (m-6-4.center) node [dot]{} -- (m-4-5.center) node [dot]{} node [midway,above,pos=.60] {} node [above=0pt] {$I_{2,j_1}$} node [right=2pt] {$\quad\cdots$};
			\draw[densely dashed] (m-6-4.center) node [dot]{} -- (m-5-5.center) node [dot]{} node [midway,below,pos=.60] {} node [below=2pt] {$I_{2,j_2}$};
			\draw[densely dashed] (m-6-4.center) node [dot]{} -- (m-8-5.center) node [dot]{} node [midway,below,pos=.60] {} node [below=0pt] {$I_{2,j_i}$} node [right=2pt] {$\quad\cdots$};
			\path (m-5-5) -- (m-8-5) node [font=\large, midway, sloped] {$\dots$};
			\draw[-] (m-5-5.center) node [dot]{} -- (m-3-6.center) node [dot]{} node [above=0pt] {$x^1$} node [right=2pt] {$\quad\cdots$};
			\draw[-] (m-5-5.center) node [dot]{} -- (m-4-6.center) node [dot]{} node [above=0pt] {$x^2$} node [right=2pt] {$\quad\cdots$};
			\draw[-] (m-5-5.center) node [dot]{} -- (m-7-6.center) node [dot]{} node [below=2pt] {$x^K$} node [right=2pt] {$\quad\cdots$};
			\path (m-4-6) -- (m-7-6) node [font=\large, midway, sloped] {$\dots$};
			\node () [rectangle, thick, rounded corners, draw, dotted, fit=(m-1-2) (m-9-2), inner xsep=3pt, inner ysep=10pt] {}; 
			\node () [rectangle, thick, rounded corners, draw, dotted, fit=(m-1-3) (m-9-3), inner xsep=3pt, inner ysep=10pt] {};
			\node () [rectangle, thick, rounded corners, draw, dotted, fit=(m-1-5) (m-9-5), inner xsep=3pt, inner ysep=10pt] {};
			\node () [rectangle, thick, rounded corners, draw, dotted, fit=(m-1-6) (m-9-6), inner xsep=3pt, inner ysep=10pt] {};
		\end{tikzpicture}
		\caption{The game tree.}
	\end{subfigure}
	\caption[Extensive for diagram]{Diagram of the game tree from Definition~\ref{definition:game_tree}. The dotted boxes denote the set of nodes where the insider (marked $X$) and the noise trader (marked $Z$) take a decision.}
	\label{figure:game_tree} 
\end{figure}%
The set of nodes where new information is revealed, and where the insider and the noise trader make a move are given by
\begin{align*}
	&\bT_V := \{ \tau \in \bT :\; \tau = r \text{ or } \tau = (I_{1, i_1}, x_1, z_1, \ldots, z_t) \text{ for } t \in \{1, \ldots, T \}\}, \\
	&\bT_X := \{ \tau \in \bT :\; \tau = I_{1, i_1} \text{ or } \tau = (I_{1, i_1}, x_1, z_1, \ldots, I_{t, i_t}) \text{ for } t \in \{2, \ldots, T \}\},\ \text{and} \\
	&\bT_Z := \{ \tau \in \bT :\; \tau = (I_{1, i_1}, x_1) \text{ or } \tau = (I_{1, i_1}, x_1, z_1, \ldots, I_{t, i_t}, x_t) \text{ for } t \in \{2, \ldots, T \}\},
\end{align*}
respectively. Next, we furnish these sets of nodes with transition probabilities, starting with the \emph{nature player}. For a finite set~$A$, the simplex~$\Delta A := \{ p \in [0,1]^A :\ \sum_{a \in A} p(a) = 1 \}$ is identified with the set of probability distributions over the points in $A$ by $p(\{a\}):=p(a)$.
\begin{definition}[Probability assignments] \label{definition:probability_assignment}
	We fix $\nu \in \Delta I$ with $\nu > 0$ and $\zeta \in \Delta E_Z$ with $\zeta > 0$.
	The \emph{probability assignment} $p_V$ assigns to every $\tau \in \bT_V$ a probability distribution over the direct successors of $\tau$ with probabilities $p_V (\tau, I_{1, i_1}) := \sum_{i \in I_{1, i_1}} \nu (\{i\})$ for $\tau = r$, and
	\begin{equation*}
		\begin{aligned}
			p_V (\tau, I_{t+1, i_{t+1}}) := \frac{\sum_{i \in I_{t+1, i_{t+1}}} \nu (\{i\})}{\sum_{j \in I_{t, i_t}} \nu (\{j\})}
		\end{aligned} \qquad
		\begin{aligned}			
			& \text{for } t \in \{1, \ldots, T\},\ I_{t+1, i_{t+1}} \subseteq I_{t, i_{t}},\\
			& \tau = (I_{1, i_1}, x_1, z_1, \ldots, I_{t, i_t}, x_t, z_t).
		\end{aligned}
	\end{equation*}
	The \emph{probability assignment} $p_Z$ is given by
	\begin{equation} \label{equation:probability_assignment_zeta}
		p_Z : \bT_Z \to \Delta E_Z, \qquad p_Z(\tau, \,\cdot\,) := \zeta.
	\end{equation}
\end{definition}
In (\ref{equation:probability_assignment_zeta}), with slight abuse of notation, we identify each direct successor of $\tau \in \bT_Z$ with an order~$z \in E_Z$.
The insider, on the other hand, observes besides the fundamental information about the true value also the noise trader's {\em past} trades before making a trading decision (the latter is a standard assumption in the Kyle model motivated by the observability of the past prices that depend on the total demand):
\begin{definition}[Behaviour strategy] \label{definition:behaviour_strategy}
	A \emph{behaviour strategy}~$\xi$ is a mapping from the insider's nodes~$\bT_X$ to the set of probability distributions over trades~$E_X$, 
	\begin{equation*}
		\xi : \bT_X \to \Delta E_X, \qquad
		\tau \mapsto \xi(\tau, \;\cdot\;).
	\end{equation*}
	With $\Xi$ we denote the set of behaviour strategies.
\end{definition}

For a behaviour strategy~$\xi \in \Xi$ and a node~$\tau \in \bT_X$, one interprets $\xi(\tau, \{x\})$ with $x>0$ ($x<0$) as the probability of buying (selling) $\abs{x}$ shares at node $\tau$. We define the set of terminal nodes by
\begin{equation*}
	\Omega := \left\{\tau \in \bT :\ \tau = (I_{1, i_1}, x_1, z_1, \ldots,\allowbreak I_{T+1, i}) \right\}.
\end{equation*}
The realisation probability~$p$ is a mapping from $\Xi$ to $\Delta \Omega$, assigning to every strategy a probability distribution over outcomes according to
\begin{equation} \label{equation:realisation_probabilities}
	\begin{split}
		p^\xi(\omega) := 
		\nu(\{i\}) \prod_{s=1}^{T} \xi((I_{1, i_1}, x_1, z_1, \ldots , I_{s, i_s}), \{x_s\}) \; \zeta(\{z_s\}) \qquad\\\qquad
		\text{for } \omega = (I_{1, i_1}, x_1, z_1, \ldots , I_{T+1, i_{}}). 
	\end{split}
\end{equation}

\subsection{Kyle equilibrium}

It is a defining feature of Kyle-type models that the market maker can not distinguish the orders of the insider and the noise trader. In each trading round, he only observes the total order flow~$y \in E_Y := \{ x + z :\; x \in E_X, z \in E_Z \}$. 
Following \cite[Equation (3.3)]{kyle_continuous_1985}, a price~$S_t$ at time~$t=1, \ldots, T$ is a function of the total order flow~$(x_1+z_1, \ldots, x_t+z_t)$.

\begin{definition}[Pricing system] \label{definition:pricing_system}
	A \emph{pricing system} $S := (S_t)_{t \in \{1,\ldots,T\}}$ of the market maker is a family of functions
	\begin{equation*} 
		S_t : (E_Y)^t  \to [v^N, v^1], \qquad
		(y_1, \ldots, y_t)  \mapsto S_t(y_1, \ldots, y_t).
	\end{equation*}
	With $\mathcal{S}$ we denote the set of pricing systems.
\end{definition}

Given a pricing system~$S \in \cS$, the payoff of the insider is defined by 
\begin{equation*}
	U(\omega, S) := \sum_{t=1}^T \left[v^i - S_t(x_1+z_1, \ldots, x_t+z_t)\right] \, x_t \quad \text{for } \omega = (I_{1, i_1}, x_1, z_1, \ldots, I_{T+1, i}),
\end{equation*}
and her objective is to maximise the expected utility
\begin{equation} \label{equation:expected_utility_discrete}
	u(\xi, S) := \sum_{\omega \in \Omega} p^\xi(\omega)\; U(\omega, S)  \longrightarrow \max_{\xi \in \Xi}!
\end{equation}
Assuming that the insider plays $\xi \in \Xi$, the joint probability of fundamental information~$i$ and order flow~$(y_1, \ldots, y_t)$ is given by
\begin{equation} \label{equation:probabilities_v_y}
	p^\xi_Y(i, y_1, \ldots, y_t):=\quad \sum_{\mathclap{\substack{(I_{1, i_1}, x_1, z_1, \ldots, I_{T+1, i}) \in \Omega\\\text{s.t. } x_s + z_s = y_s\text{ for all }s=1,\ldots,t}}} \quad p^\xi(I_{1, i_1}, x_1, z_1, \ldots, I_{t, i_t}, x_t, z_t,\ldots,  I_{T+1, i}),
\end{equation}
where $i$ uniquely determines $(I_{1, i_1}, I_{2, i_2}, \ldots, I_{T+1, i})$. 
Using Bayes' rule, the market maker can infer the conditional probability of $i$ given $(y_1, \ldots, y_t)$ assuming the insider plays $\xi$.
If $p^\xi_Y( y_1, \ldots, y_t) := \sum_{i=1}^{N} p^\xi_Y(i, y_1, \ldots, y_t) > 0$, it reads 
\begin{equation} \label{equation:probabilities_conditional}
	p^\xi_Y(i \mid  y_1, \ldots, y_t) := \frac{p^\xi_Y(i, y_1, \ldots, y_t)}{p^\xi_Y( y_1, \ldots, y_t)}.
\end{equation}
\begin{definition}[Rational pricing] \label{definition:rational_pricing_discrete}
	A pricing system $S \in \cS$ for the market maker is \emph{rational} assuming the insider plays $\xi \in \Xi$ if
	\begin{equation} \label{equation:rational_pricing_discrete}
		S_t(y_1, \ldots, y_t) = \sum_{i=1}^N p^\xi_Y(i \mid y_1, \ldots, y_t) \; v^i,
	\end{equation}
	for all $(y_1, \ldots, y_t) \in (E_Y)^t$, $t=1, \ldots, T$, with $p_Y^\xi( y_1, \ldots, y_t) > 0$, i.e, the price equals the expectation of the true value of the asset conditional on the market maker's information. 
\end{definition}

In equilibrium, the assumption of the market maker coincides with the strategy played by the insider. In addition, the insider's strategy is optimal given the prices quoted by the market maker.

\begin{definition}[Kyle equilibrium]\label{25.6.2024.2}
	A \emph{Kyle equilibrium} is a pair~$(\xi^\star, S^\star) \in \Xi \times \cS$ satisfying
	\begin{enumerate}[(i)]		
		\item \textit{Profit maximisation:} Given $S^{*}$, the strategy~$\xi^\star$ is optimal according to (\ref{equation:expected_utility_discrete}),
		\item \textit{Rational pricing:} Given $\xi^\star$, the pricing system $S^\star$ is rational according to (\ref{equation:rational_pricing_discrete}).
	\end{enumerate} 
\end{definition}
The existence of an equilibrium can be shown without further probabilistic restrictions on the true value, the insider's dynamic information, and the noise trader's actions. Example~\ref{example:no_pure_equilibrium} at the end of this section shows that equilibria need not exist in pure strategies of the insider. 
\begin{theorem} \label{theorem:discrete_trade}
	Every Kyle game has a Kyle equilibrium.
\end{theorem}

\begin{proof}
	Follows from Theorem~\ref{theorem:sequential_Kyle_equilibrium} below.
\end{proof}

\begin{remark}[Kyle model as a dynamic social system in the sense of \citet{debreu_social_1952}] \label{remark:social_system}
	We can now place the interaction of the insider and the market maker in a game theoretic context, even though the latter
	does not maximise any particular objective. Rational pricing in the sense of Definition~\ref{definition:rational_pricing_discrete}
	can be seen as a constraint on market maker's allowed actions (at each observed order flow~ $(y_1,\ldots,y_t)$) that depends on the behaviour strategy he {\em assumes} the insider chooses. If $p_Y^\xi( y_1, \ldots, y_t) > 0$, the market maker has only one choice at $(y_1,\ldots,y_t)$. In a (standard) extensive form game, an action of a player can only depend on past actions of other players but not on their behaviour strategies, i.e., not on the probabilities of all possible actions. 
	
	\cite{debreu_social_1952} extends strategic games by allowing that the choice of an action by a player can be restricted by the 
	(simultaneous) actions of the other players. 
	This idea is applied to dynamic games by \cite{butler_essays_2017}, having constraints that depend on behaviour strategies of other players. 
	The games are referred to as ``generalised''.
	We refer to \cite[Chapter~2]{butler_essays_2017} for a detailed discussion about conceptual issues, but note that the Kyle game does not fit into this framework, since there are infinitely many possible asset prices the market maker can set.
\end{remark}

\subsection{Sequential Kyle equilibrium}\label{section:sequential_equilibrium}

In extensive form games, Nash equilibria can be based on irrational plans of some players that need not be realised since 
the plans refer to nodes that are not reached in the equilibrium that they produce (``non-credible threat'').
To rule out such equilibria, Selten introduced the stronger criterion of a {\em subgame perfect equilibrium} (we refer to \cite[Section 5]{selten_reexamination_1975}): if the game can be restarted in a node, the equilibrium strategies have to be optimal also for the induced subgame. 
There are several refinements of this concept for games with imperfect information.
The most popular one is that of a {\em sequential equilibrium} introduced by \cite{kreps_sequential_1982} (for a detailed discussion we refer to \cite[Section 3.5]{gonzalez-diaz_introductory_2010}). 

In the following, we adapt this concept to Kyle games. While the insider has perfect information the market maker has not. Up to time~$t$, he only observes the total order flow~$(y_1,\ldots,y_t)$. The union of nodes at time~$t$ such that $x_1+z_1=y_1$, \ldots, $x_t+z_t=y_t$ is called an {\em information set} of the market maker. In information sets~$(y_1,\ldots,y_t)$ with $p^\xi_Y(y_1,\ldots,y_t)=0$ the conditional probability~(\ref{equation:probabilities_conditional}) is not defined. Consequently,  following \cite[Section~4]{kreps_sequential_1982}, the market maker has to form more subtle \emph{beliefs} regarding the fundamental information states.

\begin{definition}[System of beliefs] \label{definition:beliefs}
	A \emph{system of beliefs}~$\mu = (\mu_t)_{t=1,\ldots,T}$ of the market maker is a collection of probability distributions over $I$, indexed by the total order flow:
	\begin{equation*}
		\mu_t : (E_Y)^t  \to \Delta I, \qquad
		(y_1, \ldots, y_t) \mapsto \mu_t( \;\cdot\, \mid y_1, \ldots, y_t).  
	\end{equation*}
	With $\cM$ we denote the set of all systems of beliefs.
\end{definition}

\begin{definition}[Pricing with beliefs] \label{definition:rational_pricing_beliefs} 
	A system of beliefs $\mu \in \cM$ induces a pricing system~$S^\mu \in \mathcal{S}$ by
	\begin{equation*}
		S^\mu_t (y_1, \ldots, y_t) := \sum_{i=1}^N \mu_t(i \mid y_1, \ldots, y_t) \; v^i\quad\mbox{for all}\ (y_1, \ldots, y_t) \in (E_Y)^t,\ t=1, \ldots, T.
	\end{equation*}
\end{definition}

\begin{definition}[Rational beliefs] \label{definition:rational_beliefs} 
	A system of beliefs $\mu \in \cM$ is rational assuming $\xi\in \Xi$ if
	\begin{equation} \label{equation:rational_pricing_beliefs}
		\mu_t(i \mid y_1, \ldots, y_t) = p^\xi_Y(i \mid y_1, \ldots, y_t) \qquad \text{for all } i \in I
	\end{equation}
	and for all $(y_1, \ldots, y_t) \in (E_Y)^t$, $t=1, \ldots, T$, with $p_Y^\xi( y_1, \ldots, y_t) > 0$.
\end{definition}

For information sets~$(y_1,\ldots,y_t)$ with $p_Y^\xi( y_1, \ldots, y_t) > 0$ the beliefs~$\mu_t(i \mid y_1, \ldots, y_t)$ are uniquely determined by $\xi$ according to (\ref{equation:rational_pricing_beliefs}). In information sets that cannot be realised, \cite[Section~5]{kreps_sequential_1982} argue that rational beliefs should be the limit of conditional probabilities that result from an approximating sequence of strategies
where every node is reached with positive probability. For the Kyle game, this leads to the following definitions.

\begin{definition}[Completely mixed behaviour strategies] \label{definition:completely_mixed}
	The set~$\Xi^0$ of completely mixed behaviour strategies is given by
	\begin{equation*}
		\Xi^0 := \{ \xi \in \Xi :\; \xi(\tau,\{x\}) > 0 \text{ for all } \tau \in \bT_X, x \in E_X \}.
	\end{equation*}
\end{definition}

\begin{definition}[Consistent beliefs] \label{definition:consistent}
	A system of beliefs~$\mu \in \cM$ is \emph{consistent} with $\xi \in \Xi$ if there is an approximating sequence~$(\xi^n, \mu^n)_{n \in\bN} \subseteq \Xi \times \cM$ with $(\xi^n, \mu^n) \to (\xi, \mu)$ as $n\to\infty$ such that 
	\begin{enumerate}[(i)]
		\item $\xi^n$ is completely mixed for all $n\in\bN$ (see Definition \ref{definition:completely_mixed}),
		\item $\mu^n$ is rational assuming $\xi^n$ for all $n\in\bN$ (see Definition \ref{definition:rational_beliefs}). 
	\end{enumerate}
	(Elements of $\Xi\times\cM$ can be identified with elements of $\bR^d$ for some $d$, and convergence is understood accordingly.)	
\end{definition}
The insider is the only player who appears in the game tree. Since she has perfect information, she need not have beliefs, and at any node a subgame can be started. For notational convenience, we only start subgames at $\tau\in\bT_V$ and do this in the following way: 

For $\tau=r$, the subgame is the game itself. For $\tau = (I_{1, i_1}, x_1, z_1, \ldots, I_{t, i_t}, x_t, z_t) \in \bT_V$, $t\in\{1,\ldots,T\}$, we restrict the game tree to the subtree with root node~$\tau$ (which leads to a restriction of terminal nodes to those that come after $\tau$), and realisation probabilities~$p_\tau^\xi$ conditional on starting in $\tau$, analogously to (\ref{equation:realisation_probabilities}), are given by
\begin{equation} \label{equation:realisation_probabilities_subgame}
	\begin{split}
		p_\tau^\xi(\omega) := \frac{\nu(\{i\})}{\sum_{j \in I_{t, i_t}} \nu(\{j\})} \prod_{s={t+1}}^T \xi((I_{1, i_1}, x_1, z_1, \ldots , I_{s, i_s}), \{x_s\}) \; \zeta(\{z_s\}) \qquad\\\qquad
		\text{for } \omega = (I_{1, i_1}, x_1, z_1, \ldots, I_{t, i_t}, x_t, z_t, I_{t+1, i_{t+1}}, x_{t+1}, z_{t+1},\ldots,  I_{T+1, i}).
	\end{split}
\end{equation}
The expected utility for the subgame starting in $\tau \in \bT_V$ becomes, analogously to (\ref{equation:expected_utility_discrete}),
\begin{equation*}
	u_\tau(\xi, S) := \sum_{\substack{\omega \in \Omega\\\omega \text{ after } \tau}} p_\tau^\xi(\omega)\; U(\omega, S).
\end{equation*}

\begin{definition}[Subgame optimality] \label{definition:subgame_optimal}
	A behaviour strategy~$\xi^\star \in \Xi$ is \emph{subgame optimal} given $S \in \cS$ if at every node $\tau \in \bT_V$,
	\begin{equation} \label{equation:subgame_optimal}
		\xi^\star \in \argmax_{\xi \in \Xi} u_\tau(\xi, S).
	\end{equation}
\end{definition}

Subgame optimality corresponds to \emph{sequential rationality} in \cite[Section~4]{kreps_sequential_1982} (not to be confused with rational pricing). Because the insider has perfect information, subgame optimality is equivalent to the notion of \emph{subgame perfection} in 
\cite[Section~5]{selten_reexamination_1975}.

\begin{definition}[Sequential Kyle equilibrium]
	A \emph{sequential Kyle equilibrium} is a pair~$(\xi^\star, \mu^\star) \in \Xi \times \cM$ satisfying:
	\begin{enumerate}[(i)]
		\item Profit maximisation: Given $S^{\mu^\star}$, the strategy~$\xi^\star$ is subgame optimal (see Definition \ref{definition:subgame_optimal}),
		\item Rational pricing: The system of beliefs~$\mu^\star$ is consistent with $\xi^\star$ (see Definition \ref{definition:consistent}).	
	\end{enumerate} 
\end{definition}

Whereas \cite[Equations (3.5), (3.7)]{kyle_continuous_1985} already envisions a form of sequential equilibrium, the novel aspect is the 
formal relation to \cite{kreps_sequential_1982} with the
consistency condition regarding the beliefs of the market maker that generalises conditional probabilities derived from Bayes' rule and
allows the case $p^\xi_Y(y_1, \ldots, y_t)=0$ to be handled. The main result of this section is the following.

\begin{theorem} \label{theorem:sequential_Kyle_equilibrium}
	Every Kyle game has a sequential Kyle equilibrium.
\end{theorem}

\subsection{Proof of Theorem~\ref{theorem:sequential_Kyle_equilibrium}}

The proof of Theorem~\ref{theorem:sequential_Kyle_equilibrium} is based on the classical ideas that certain extensive form games have a \emph{(trembling hand) perfect equilibrium}, see \cite[Section 11]{selten_reexamination_1975}, and every perfect equilibrium is a sequential equilibrium, see \cite[Proposition 1]{kreps_sequential_1982}. However, we have to adapt this approach to accommodate the structure of a Kyle game. The proof we present is self-contained, without formally introducing Selten's \emph{agent normal form}, only requiring the reader to be familiar with Kakutani's fixed point theorem (see, e.g., \cite[Theorem~2.2.1]{gonzalez-diaz_introductory_2010}).

For $\epsilon > 0$ small enough such that $\epsilon\abs{E_X}<1$, the \emph{$\epsilon$-perturbed game} is the Kyle game with insider strategies in
\begin{equation*}
	\Xi^\epsilon :=\{ \xi \in \Xi :\; \xi(\tau,\cdot)\in\Delta E_X^\epsilon \text{ for all } \tau \in \bT_X\},
\end{equation*}
where $\Delta E_X^\epsilon := \{ p \in \Delta E_X: p(\{x\}) \ge \epsilon \text{ for all } x \in E_X \}$. 
We write $\xi_\tau := \xi(\tau, \,\cdot\;)$ and follow the usual notation that 
\begin{equation*}
	\xi = (\xi_{-\tau}, \xi_{\tau}) \qquad \text{where}\qquad	\xi_{-\tau} :  \bT_X \setminus \{\tau\} \to \Delta E_X,\quad \xi_{-\tau}(\tau') = \xi_{\tau'}.
\end{equation*}
For a node~$\tau := (I_{1, i_1}, x_1, z_1, \ldots, I_{t, i_t})\in\bT_X$ we define, analogously to (\ref{equation:realisation_probabilities_subgame}), 
\begin{equation*}
	\begin{split}
		p_\tau^\xi(\omega) := \frac{\nu(\{i\})}{\sum_{j \in I_{t, i_t}} \nu(\{j\})} \prod_{s=t}^T \xi((I_{1, i_1}, x_1, z_1, \ldots , I_{s, i_s}), \{x_s\}) \; \zeta(\{z_s\}) \qquad\\\qquad
		\text{for } \omega = (I_{1, i_1}, x_1, z_1, \ldots, I_{t, i_t}, x_t, z_t, I_{t+1, i_{t+1}}, x_{t+1}, z_{t+1},\ldots,  I_{T+1, i})
	\end{split}
\end{equation*}
and
\begin{equation*}
	u_\tau(\xi, S) := \sum_{\substack{\omega \in \Omega\\\omega \text{ after } \tau}} p_\tau^\xi(\omega)\; U(\omega, S).
\end{equation*}

The difference of the perturbed game to the original game is that every
possible trade of the insider has to be chosen with a minimum probability~$\epsilon>0$. 
For small $\epsilon$, the effect of this constraint on the maximiser is small. 
On the other hand, under $\xi\in\Xi^\epsilon$ with $\epsilon>0$, every node in the game tree is reached with positive probability, and rational beliefs of the market maker assuming $\xi$ are given by Bayes' rule for conditional probabilities. 
Having maximisers for the perturbed games, we let $\epsilon$ tend
to zero and show that there is a limiting system of beliefs, and a limiting strategy that maximises every subproblem of the original game.
By definition, this means that the limiting system of beliefs is consistent with the limiting 
strategy~(cf. Definition~\ref{definition:consistent}).

To work out this idea in detail, we next define the self-correspondence
\begin{equation}\label{25.6.2024.1}
	\begin{gathered}
		\begin{aligned}
			F^\epsilon : \Xi^\epsilon \times \cM &\rightrightarrows \Xi^\epsilon \times \cM, \\
			(\xi, \mu) &\mapsto  \prod_{\tau \in \bT_X} f^\epsilon_\tau(\xi, \mu) \times f_0^\epsilon(\xi), \qquad\text{where}
		\end{aligned}
		\\
		f_\tau^\epsilon(\xi, \mu) := \argmax_{p\in \Delta E_X^\epsilon} \; u_\tau((\xi_{-\tau},p), S^\mu) \ \text{ and }\
		f_0^\epsilon(\xi) := \{\mu' \in \cM: \mu' \text{ consistent with } \xi \}.
	\end{gathered}
\end{equation}
The first component of the correspondence~$F^\epsilon$ maps a pair~$(\xi, \mu)\in\Xi^\epsilon \times \cM$ to the set of ``locally optimal'' strategies in $\Xi^\epsilon$ given the pricing system~$S^\mu$ and the insider's strategy~$\xi$ in all other nodes. The second component maps $(\xi, \mu)$ to the set of beliefs consistent with $\xi$. The intuition behind the definition of $F^\epsilon$ is that by the dynamic programming principle, a strategy is optimal if and only if it is locally optimal at each node given the strategy at all other nodes.
This allows to reduce the $\epsilon$-perturbed game to a static game with different players at each node (called ``agents'').

\begin{lemma} \label{lemma:fixed_point}
	For every $\epsilon\in(0,1/|E_X|)$, $F^\epsilon$ has a fixed point.
\end{lemma}
Such a fixed point can be seen as an equilibrium of the $\epsilon$-perturbed Kyle game when the insider maximises her utility by 
choosing trades separately at each node and considering the choices at other nodes as given. 
\begin{proof}[Proof of Lemma~\ref{lemma:fixed_point}]
	We fix $\epsilon\in(0,1/|E_X|)$. One can identify $\Xi$ and $\cM$ with subsets of $\bR^{d}$ for some $d$, and topological properties are to be understood accordingly. In order to apply Kakutani's fixed point theorem, we need to verify that $F^\epsilon$ is upper hemicontinuous, non-empty-, closed-, and convex-valued (see, e.g., \cite[Theorem~2.2.1]{gonzalez-diaz_introductory_2010}). It is sufficient to prove these properties for the components $f_\tau^\epsilon$, $f_0^\epsilon$.
	
	\emph{Step 1.} Fix $\tau \in \bT_X$ and consider $f_\tau^\epsilon: \Xi^\epsilon \times \cM \rightrightarrows \Delta E_X^\epsilon$. Any element of $\Delta E_X^\epsilon$ is a convex combination of the \emph{pure strategies} $\delta_x$, $x \in E_X$. Since (\ref{equation:realisation_probabilities})
	is linear in the transition distribution at a single node and by (\ref{equation:expected_utility_discrete}), 
	the set  $f_\tau^\epsilon(\xi, \mu)$ consists of those convex combinations where suboptimal pure strategies receive only the minimal weight $\epsilon$, i.e., 
	\begin{equation} \label{equation:purification}
		\begin{split}
			f_\tau^\epsilon(\xi, \mu) = \{ p\in \Delta E_X^\epsilon :\; p(\{x\}) = \epsilon\ \mbox{for all }x\in E_X\setminus\widehat{E}_X(\xi,\mu,\tau)\},
		\end{split}
	\end{equation}
	where
	\begin{equation*}
		\widehat{E}_X(\xi,\mu,\tau):=\{x\in E_X : 
		u_\tau((\xi_{-\tau}, \delta_x), S^\mu) \ge u_\tau((\xi_{-\tau}, \delta_{x'}),S^\mu)\ \mbox{for all }x'\in E_X\}.
	\end{equation*}
	Consequently, the set~$f_\tau^\epsilon(\xi, \mu)$ is obviously closed, convex, and non-empty. It remains to show upper hemicontinuity (see, e.g., \cite[page~21]{gonzalez-diaz_introductory_2010} for a definition). 
	The conditional probabilities $p^\xi_\tau$ in (\ref{equation:realisation_probabilities_subgame}) are continuous in $\xi$. As the concatenation of continuous functions, $(\xi,\mu) \mapsto u_\tau((\xi_{-\tau}, \delta_x), S^\mu) $ is continuous for all $x\in E_X$. 
	Thus, $\widehat{E}_X(\xi',\mu',\tau)\subseteq \widehat{E}_X(\xi,\mu,\tau)$ for all $(\xi',\mu')$ in a neighborhood of $(\xi,\mu)$
	since a suboptimal pure strategy remains suboptimal under a slight perturbation of $(\xi, \mu)$ and there are only finitely many of them. In conjunction with (\ref{equation:purification}), we obtain $f_\tau^\epsilon(\xi', \mu')\subseteq f_\tau^\epsilon(\xi, \mu)$, which implies upper hemicontinuity.
	
	\emph{Step 2.} Now we turn to $f_0^\epsilon : \Xi^\epsilon \rightrightarrows \cM$. Any strategy in $\xi \in \Xi^\epsilon$ is completely mixed, and thus has exactly one system of beliefs that is consistent with $\xi$, namely the conditional probabilities $\mu := p^\xi_Y$ from (\ref{equation:probabilities_conditional}), where the denominator is strictly positive for all $(y_1,\ldots,y_t)$.
	The singleton $\{\mu\}$ is obviously non-empty, closed and convex. From (\ref{equation:probabilities_conditional}) and strict positivity it follows that $\mu$ considered as a function of $\xi$ is continuous, and thus $f_0^\epsilon$ is upper hemicontinuous.
	
	Applying Kakutani's fixed point theorem to $F^\epsilon$ one obtains a fixed point.
\end{proof}

\begin{lemma} \label{lemma:subgame_optimal}
	For $\epsilon > 0$ let $(\xi^\epsilon, \mu^\epsilon)$ be a fixed point of $F^\epsilon$. Then the strategy $\xi^\epsilon$ is subgame optimal in $\Xi^\epsilon$ given $S^{\mu^\epsilon}$.
\end{lemma}
\begin{proof}
	Let $(\xi^\epsilon, \mu^\epsilon)$ be a fixed point of $F^\epsilon$ according to Lemma~\ref{lemma:fixed_point}. As such, $\xi^\epsilon$ is locally optimal given $S:=S^{\mu^\epsilon}$, i.e.,
	\begin{equation}\label{equation:lemma_subgame_optimal}
		\xi^\epsilon_\tau \in \argmax_{p\in \Delta E_X^\epsilon} \; u_\tau\big((\xi_{-\tau}^\epsilon, p), S\big) \quad \text{for all } \tau \in \bT_X.
	\end{equation}
	We prove (\ref{equation:subgame_optimal}) by backward induction over the period~$t$ in which the node~$\tau\in\bT_V$ lies. This means, we have to show that
	\begin{equation*}
		u_\tau(\xi^\epsilon, S) \ge u_\tau(\xi, S)\qquad\text{for all }\tau=(I_{1, i_1}, x_1, z_1, \ldots, I_{t-1, i_{t-1}}, x_{t-1}, z_{t-1}),\ \xi\in\Xi^\epsilon,
	\end{equation*}
	using that the assertion holds for $t$ instead of $t-1$ (the base case $t-1=T$ follows by the same arguments as below). Making use of the special structure of the Kyle game, the proof is shorter than the original one by Selten for general extensive form games (see \cite[Lemma~6]{selten_reexamination_1975}).
	
	{\em Step 1.} Throughout the proof, we fix a competing strategy~$\xi\in\Xi^\epsilon$. In the first step, we consider nodes of the form 
	$\tau':=(I_{1, i_1}, x_1, z_1, \ldots, I_{t,i_t})\in\bT_X$, i.e., nodes in which new insider information is already revealed, 
	and define the strategy $\xi':=(\xi^\epsilon_{-\tau'},\xi_{\tau'})$. From (\ref{equation:lemma_subgame_optimal}) it follows that
	\begin{equation*}
		u_{\tau'}(\xi^\epsilon, S)\ge u_{\tau'}(\xi', S).
	\end{equation*}  
	On the other hand, by definition of $u_{\tau'}$ and $u_{\tau'_{k,l}}$, one has
	\begin{equation}\label{equation:decomposition_tau'}
		u_{\tau'}(\xi', S) = \sum_{k=1}^K\sum_{l=1}^L \xi((I_{1, i_1}, x_1, z_1, \ldots , I_{t, i_t}),\{x^k\}) \; \zeta(\{z^l\}) \; u_{\tau'_{k,l}}(\xi^\epsilon, S)
	\end{equation}
	where $\tau'_{k,l}:=(I_{1, i_1}, x_1, z_1, \ldots, I_{t,i_t}, x^k,z^l)\in\bT_V$. From the induction hypothesis it follows that the RHS of (\ref{equation:decomposition_tau'}) dominates the LHS of
	\begin{equation*}
		\sum_{k=1}^K\sum_{l=1}^L \xi((I_{1, i_1}, x_1, z_1, \ldots , I_{t, i_t}),\{x^k\}) \; \zeta(\{z^l\}) \; u_{\tau'_{k,l}}(\xi, S)
		= u_{\tau'}(\xi, S).
	\end{equation*}
	Put together, we arrive at $u_{\tau'}(\xi^\epsilon, S)\ge u_{\tau'}(\xi, S)$. 
	
	{\em Step 2.} For nodes $\tau:=(I_{1, i_1}, x_1, z_1, \ldots, I_{t-1,i_{t-1}}, x_{t-1},z_{t-1})\in\bT_V$, the estimate $u_\tau(\xi^\epsilon, S)\ge u_\tau(\xi, S)$ follows from Step~1 by weighting $u_{\tau'}(\xi^\epsilon, S)$ and $u_{\tau'}(\xi, S)$ by the transition probabilities in the insider's private information tree, which are exogenous.
\end{proof}
The proof of Theorem~\ref{theorem:sequential_Kyle_equilibrium} is completed by the following lemma.
\begin{lemma}
	There exists a sequence $(\epsilon_n)_{n \in \bN}\subseteq\bR_+$ with $\epsilon_n \downarrow 0$, 
	a sequence $(\xi^{\epsilon_n}, \mu^{\epsilon_n})_{n \in \bN}$ with $(\xi^{\epsilon_n}, \mu^{\epsilon_n})\in\Xi^{\epsilon_n}\times\cM$, 
	and $(\xi^{*}, \mu^{*})\in\Xi\times\cM$ such that $(\xi^{\epsilon_n}, \mu^{\epsilon_n})$ is a fixed point of 
	$F^{\epsilon_n}$ for all $n\in\bN$ and $(\xi^{\epsilon_n}, \mu^{\epsilon_n}) \to (\xi^{*}, \mu^{*})$ as $n \to \infty$. In addition, $(\xi^{*}, \mu^{*})$ is a sequential Kyle equilibrium.	
	(Elements of $\Xi\times\cM$ can be identified with elements of $\bR^d$ for some $d$, and convergence is understood accordingly.)	
\end{lemma}

\begin{proof}
	For every $n\in\bN$, let $(\xi^{\sfrac{1}{n}}, \mu^{\sfrac{1}{n}})\in\Xi^{\sfrac{1}{n}}\times\cM$ be a fixed point of $F^{\sfrac{1}{n}}$ provided by Lemma~\ref{lemma:fixed_point}. Since $\xi^{\sfrac{1}{n}}(\tau,\{x\})\in [0,1]$ for all $\tau\in\bT_X$ and $x\in E_X$, there exists a  
	subsequence~$(n_k)_{k \in \bN}$ such that $\xi^{\sfrac{1}{n_k}} \to \xi^\star$ as $k \to \infty$ for some $\xi^\star$ with $\xi^\star(\tau,\{x^i\})\ge 0$, $i=1,\ldots,K$ and $\sum_{i=1}^K \xi^\star(\tau,\{x^i\})=1$. Because $1/{n_k} > 0$ for all $k\in\bN$, every $\mu^{\sfrac{1}{n_k}}$ is the unique system of beliefs that is rational assuming $\xi^{\sfrac{1}{n_k}}$. 
	There exists a further subsequence $(n_{k_j})_{j \in \bN}$ and $\mu^\star\in\cM$ such that $\smash{\mu^{1/n_{k_j}}} \to \mu^\star$ as $j \to \infty$. In conclusion, $\smash{(\xi^{\sfrac{1}{n_{k_j}}}, \mu^{\sfrac{1}{n_{k_j}}})} \to (\xi^\star, \mu^\star)$ as $j \to \infty$, $\smash{\xi^{\sfrac{1}{n_{k_j}}}}$ is completely mixed, and $\smash{\mu^{\sfrac{1}{n_{k_j}}}}$ is rational assuming $\xi^{n_{k_j}}$, so $\mu^\star$ is consistent with $\xi^\star$ (see Definition \ref{definition:consistent}).
	
	It remains to show that $\xi^\star$ is subgame optimal given $S^{\mu^\star}$. W.l.o.g. $\smash{n_{k_j}} = n$ for all $j \in \bN$. 
	We fix a node $\tau \in \bT_V$ and a competing strategy~$\widetilde{\xi}\in\Xi$. We construct an approximating sequence~$(\widetilde{\xi}^n)_{n\in\bN}$ with $\widetilde{\xi}^n\in\Xi^{\sfrac{1}{n}}$ by $\widetilde{\xi}^n(\tau,\{x\}):= {\sfrac{1}{n}} + \widetilde{\xi}(\tau,\{x\})(1-\abs{E_X}/n)$ for all $n>\abs{E_X}$. 
	Following Lemma~\ref{lemma:subgame_optimal} the strategy $\xi^{\sfrac{1}{n}}$ is subgame optimal in $\Xi^{\sfrac{1}{n}}$ given $S^{\mu^{\sfrac{1}{n}}}$ which implies 
	\begin{equation*}
		u_\tau\big(\xi^{\sfrac{1}{n}}, S^{\mu^{\sfrac{1}{n}}}\big) \geq u_\tau \big(\widetilde{\xi}^n, S^{\mu^{\sfrac{1}{n}}}\big)\quad\mbox{for all }n>|E_X|.
	\end{equation*}
	Because $(\xi, \mu) \mapsto u_\tau(\xi, S^\mu)$ is jointly continuous, $(\xi^{\sfrac{1}{n}}, \mu^{\sfrac{1}{n}}) \to (\xi^\star, \mu^\star)$, and $(\widetilde{\xi}^n, \mu^{\sfrac{1}{n}}) \to (\widetilde{\xi}, \mu^\star)$ as $n\to\infty$, we can pass to the limit $n \to \infty$ on both sides to get
	\begin{equation*}
		u_\tau \big(\xi^\star, S^{\mu^\star}\big) \geq u_\tau \big(\widetilde{\xi}, S^{\mu^\star}\big).
	\end{equation*}
\end{proof}

\begin{remark}[Semi-explicit characterisation of the set of Kyle equilibria]\label{25.6.2024.3}
	The specification of the Kyle model as a game has the benefit that equilibria can be characterised as fixed points of a certain 
	self-correspondence of best replies, which enables a systematic study of Kyle equilibria. 
The self-correspondence we need is a variant of the self-correspondence~$F^\epsilon$ from (\ref{25.6.2024.1}) that is used to get a {\em sequential} Kyle equilibrium. The variant acts on the set of mixed strategies
denoted by $\Xi^m$ (that are lotteries over behaviour strategies with Dirac measures). Analogous to (\ref{equation:realisation_probabilities}),
a mixed strategy induces realisation probabilities on the tree, which we use to define the expected utility of a mixed strategy and rational prices assuming a mixed strategy. Then, the new self-correspondence reads	
	\begin{equation}\label{27.6.2024.1}
		\begin{aligned}
			F^m : \Xi^m \times \mathcal{S} &\rightrightarrows \Xi^m \times \mathcal{S}, \\
			(\xi, S) &\mapsto \argmax_{\wt{\xi} \in \Xi^m} \; u(\wt{\xi}, S) \times \{\wt{S} \in \mathcal{S}: \wt{S} \text{ is rational assuming } \xi \}.
		\end{aligned}
	\end{equation}
Since $F^m$ acts on the set of mixed strategies, we have to convert mixed strategies into behaviour strategies and vice versa. 	
First, we start with a fixed point of $F^m$ denoted by $(\xi,S)\in \Xi^m \times \mathcal{S}$.
Because the insider has perfect information, the game has \emph{perfect recall}, and Kuhn's theorem applies: for every mixed strategy there exists a realisation equivalent behaviour strategy (see, e.g., \cite[Theorem~3.2.1]{gonzalez-diaz_introductory_2010}). The pair consisting of a behaviour strategy which is realisation equivalent to $\xi$ and the price function~$S$ is obviously a Kyle equilibrium in the sense of Definition~\ref{25.6.2024.2}.
Second, starting with a Kyle equilibrium in the sense of Definition~\ref{25.6.2024.2}, the insider's behaviour strategy induces an ``equivalent'' mixed strategy~$\xi\in \Xi^m$ (see, e.g., \cite[Definition~3.2.3]{gonzalez-diaz_introductory_2010}). 
Since the realisation probabilities on the tree are the same, the mixed strategy and the price function have to be a fixed point of $F^m$.
This means that by (\ref{27.6.2024.1}) we get all Kyle equilibria.
\end{remark}

\begin{example}[$\nexists$ equilibrium in pure strategies] \label{example:no_pure_equilibrium}
We provide a minimalist example of a two-period Kyle game in which no equilibrium  
with a pure strategy of the insider exists. There are multiple equilibria, but there is one node at which the insider must randomise in equilibrium.
Let $T=2$,  $E_V=\{0,1\}$, $\mathcal{I}_1=\mathcal{I}_2=\{\{0\},\{1\}\}$, i.e., the true value is completely revealed to the insider at the very beginning, $\nu(\{0\})=\nu(\{1\})=1/2$, $E_Z=\{-1,0,1\}$, $\zeta(\{-1\})=\zeta(\{1\})=\eps$, $\zeta(\{0\})=1-2\eps$, and $E_X=\{0,1\}$. The parameter~$\eps>0$ is not yet specified, but it should be ``small''. This means that the total demand coincides with the insider's demand with high probability (one can show numerically that the assertion holds for any $\eps \in (0, 0.19371\ldots)$).

By definition, prices have to lie between $0$ and $1$. Therefore, in the second period it is always optimal for the insider to buy if $v=1$ and to do nothing if $v=0$. Let us show that no pure strategy can be part of an equilibrium: if such a strategy is assumed by the market maker, he 
sets prices under which the strategy is not optimal for the insider. This is obvious for all but one pure strategy: the only interesting strategy is 
\begin{equation}\label{10.6.2024.1}
	\xi(\tau,\{1\})=1\ \text{for every node~$\tau$ after $v=1$}, \qquad {and} \qquad \xi(\tau,\{0\})=1\ \text{after $v=0$},
\end{equation}
i.e., for $v=1$ the insider already buys in the first period with probability~$1$. For all other strategies we refrain from writing down that they are not optimal given the associated rational prices (e.g., if the market maker assumes that the insider does not trade in the first period, he sets the price constant to $1/2$ in the first period, which gives the insider an incentive to buy, and so it does not lead to an equilibrium). 

Now, let $S$ be the pricing system that is rational assuming the insider plays~$\xi$ from (\ref{10.6.2024.1}). Let $\wt{\xi}$ be the alternative pure strategy where the insider does not trade in the first period: we set $\wt{\xi}(1,\{0\})=1$ and $\wt{\xi}=\xi$ for all other nodes. We want to show that 
\begin{equation}\label{11.6.2024.01}
	u(\wt{\xi},S)>u(\xi,S)\quad \mbox{for\ }\eps\ \mbox{small enough (thus $\nexists$ equilibrium in pure strategies).}  
\end{equation}
For this, it is sufficient to look at the asymptotic behaviour for $\eps\downarrow 0$. First, we consider $u(\xi,S)$. The insider mainly pays prices
$S_1(1)=1-O(\eps)$ and $S_2(1,1)=1-O(\eps)$ for the asset, whereas lower prices she only gets with probability~$O(\eps)$, and thus we obtain 
\begin{equation*}
	u(\xi,S)=O(\eps).
\end{equation*}
On the other hand, strategy~$\wt{\xi}$ buys at price~$S_2(0,1)$ with probability~$1-O(\eps)$. Assuming $\xi$, the 
state~$(y_1,y_2)=(0,1)$ is realized by two nodes with positive probability: $(v, x_1, z_1, x_2, z_2)=(0,0,0,0,1)$ and $(v, x_1, z_1, x_2, z_2)=(1,1,-1,1,0)$ (indicated with $\star$ in Figure~\ref{figure:no_pure_equilibrium}(a)).  
Both nodes have probability~$(1-2\eps)\eps/2$, and thus the rational price is given by $S_2(0,1)=1/2$. We arrive at $u(\wt{\xi},S)=1/2 + O(\eps)$ and finally (\ref{11.6.2024.01}).

The economic interpretation is as follows. If the noise trader's activity is small, the insider can only make small gains by doing the expected strategy $\xi$. 
But by buying only in the second period when $v=1$ occurs, the strategy $\wt\xi$, she irritates the market maker and gets the same price as if she were a noise trader.

Symbolic computations show that there exists an ``essentially unique'' mixed equilibrium for any $\eps \in (0, 0.19371\ldots)$: Let $\xi_\alpha$ be the mixed strategy where the insider buys with probability $\alpha$ in the first period when $v=1$, i.e. $\xi_\alpha(1,\{1\}) = \alpha = 1-\xi(1,\{0\})$ and $\xi_\alpha=\xi$ for all other nodes, there $\xi$ is optimal for any pricing system (an indifference case occurs at node~$(v,x_1,z_1)=(1,1,1)$, since in equilibrium $S_2(2,y_2)$ must be $1$ for all $y_2\in\{0,1,2\}$ if $\alpha>0$).
For fixed $\eps$ one can calculate the rational prices and associated profits as a function of $\alpha$. We have a mixed equilibrium if and only if the insider is indifferent between buying and doing nothing in the first period after $v=1$ (cf. the arguments in Step 1 of the proof of Lemma~\ref{lemma:fixed_point}, and see Figure~\ref{figure:no_pure_equilibrium}(b)). Computing equilibrium values of $\alpha$ boils down to finding the roots of a polynomial of degree seven. E.g. for $\eps=1/8$, we have to solve
\begin{equation*}
	\begin{split}
	750000\alpha^7-9485000\alpha^6+36365625\alpha^5-48108800\alpha^4-25782575\alpha^3 \qquad\\\qquad+80831674\alpha^2-21705040\alpha-10602816 = 0,
	\end{split}
\end{equation*}
and there is exactly one positive real root $<1$, $\alpha^\star=0.77464\ldots$
Putting together, $(\xi_{\alpha^\star},S)$ is a Kyle equilibrium where $S$ is given by (\ref{equation:rational_pricing_discrete}), 
excluding $S_2(2,-1)$ which can be chosen arbitrarily from $[0,1]$. To obtain all equilibria, one can replace 
$\xi_{\alpha^\star}((1,1,1),\{1\})=1$ by an arbitrary $\beta\in[0,1]$ (for $\beta<1$, $S_2(2,-1)$ must be $1$, and for $\beta=0$, $S_2(2,2)$ becomes arbitrary).
\end{example}

\noindent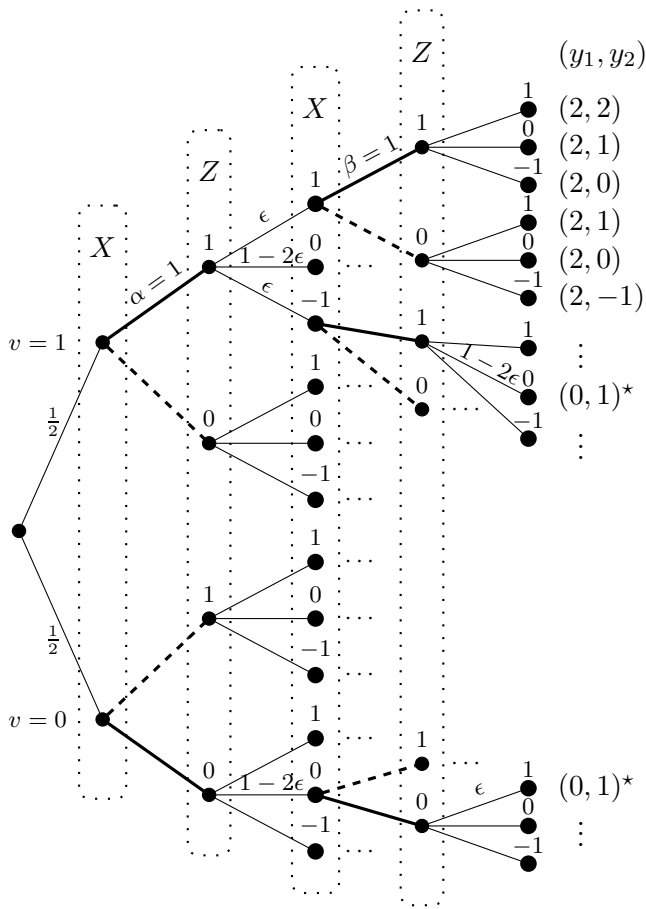
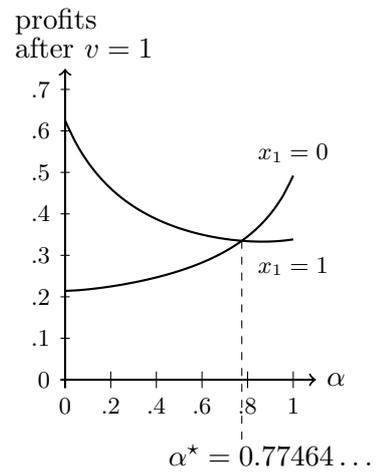
\begin{figure}[]%
	\tikzstyle{dot} = [draw, solid, fill, circle, inner sep=0.0cm, outer sep=0cm, minimum size=5pt]
	\tikzstyle{densely dashed}=[dash pattern=on 6pt off 1.5pt]
	\tikzstyle{every node}=[font=\footnotesize]
	\tikzstyle{normal} = [thin, solid]
	\tikzstyle{emph} = [very thick, solid]
	\tikzstyle{level 1}=[level distance=1.1cm, sibling distance=5cm]
	\tikzstyle{level 2}=[level distance=1.4cm, sibling distance=2cm]
	\tikzstyle{level 3}=[level distance=1.4cm, sibling distance=.75cm]
	\tikzstyle{level 4}=[level distance=1.4cm, sibling distance=1.5cm]
	\tikzstyle{level 5}=[level distance=1.4cm, sibling distance=.5cm]
	\begin{subfigure}[b]{0.65\textwidth} \centering
		\begin{tikzpicture}[grow=right]
			\node[dot] (r) {}
			child {
				node[dot, label={[left=.25cm]left: {$v=0$}}] (v0) {} 
				child {
					node[dot, label={[above=0pt] {$0$}}] (x1bottom) {}
					child {
						node[dot, label={[above=0pt] {$-1$}}] (z1bottom) {} node[right=.25cm] {$\cdots$}
						edge from parent[normal]
					}
					child {
						node[dot, label={[above=0pt] {$0$}}] () {}					
						child {
							node[dot, above=.25cm, label={[above=0pt] {$0$}}] (x2bottom) {}
							child {
								node[dot, label={[above=-3pt] {$-1$}}] () {} node[right=.25cm] {\normalsize $ $}
								edge from parent[normal]
							}
							child {
								node[dot, label={[above=-2pt] {$0$}}] () {} node[right=.5cm] {\normalsize $\vdots$}
								edge from parent[normal]
							}
							child {
								node[dot, label={[above=-2pt] {$1$}}] () {} node[right=.25cm] {\normalsize $(0,1)^\star$}
								edge from parent[normal] node[pos=.6, above, sloped] {\footnotesize$\epsilon$}
							}
							edge from parent[emph]
						}
						child {
							node[dot, below=.25cm, label={[above=0pt] {$1$}}] (cdots1) {} 
							edge from parent[emph, dashed]
						}
						edge from parent[normal] node[pos=.6, above, sloped, above=-3pt] {\footnotesize $1-2\epsilon$}
					}
					child {
						node[dot, label={[above] {$1$}}] () {} node[right=.25cm] {$\cdots$}
						edge from parent[normal]
					}
					edge from parent[emph] 
				}
				child {
					node[dot, above=.25cm, label={[above=0pt] {$1$}}] () {}
					child {
						node[dot, label={[above=0pt] {$-1$}}] () {} node[right=.25cm] {$\cdots$}
						edge from parent[normal]
					}
					child {
						node[dot, label={[above=0pt] {$0$}}] () {} node[right=.25cm] {$\cdots$}
						edge from parent[normal]
					}
					child {
						node[dot, label={[above=0pt] {$1$}}] () {} node[right=.25cm] {$\cdots$}
						edge from parent[normal]
					}
					edge from parent[emph, dashed] 
				}
				edge from parent[normal] node[pos=.4, below=3pt] {$\frac{1}{2}$}
			}
			child {
				node[dot, label={[left=.25cm]left: {$v=1$}}] (v1) {}
				child {
					node[dot, below=.25cm, label={[above=0pt] {$0$}}] () {}
					child {
						node[dot, label={[above=0pt] {$-1$}}] () {} node[right=.25cm] {$\cdots$}
						edge from parent[normal]
					}
					child {
						node[dot, label={[above=0pt] {$0$}}] () {} node[right=.25cm] {$\cdots$}
						edge from parent[normal]
					}
					child {
						node[dot, label={[above=0pt] {$1$}}] () {} node[right=.25cm] {$\cdots$}
						edge from parent[normal]
					}
					edge from parent[emph, dashed] node[pos=.5, above, sloped] {$ $}
				}
				child {
					node[dot, label={[above=0pt] {$1$}}] (x1top) {}
					child {
						node[dot, label={[above=-1pt] {$-1$}}] () {}
						child {
							node [dot, below=.3cm, label={[above=0pt] {$0$}}] (cdots) {}
							edge from parent[emph, dashed]
						}
						child {
							node[dot, below=.9cm, label={[above=0pt] {$1$}}] () {} 
							child {
								node[dot, below=.7cm, label={[above=-3pt] {$-1$}}] (00) {} node[right=.5cm of 00.center] {\normalsize $\vdots$}
								edge from parent[normal]
							}
							child {
								node[dot, below=.65cm, label={[above=-2pt] {$0$}}] (01) {} node[right=.25cm of 01.center] {\normalsize $(0,1)^\star$}
								edge from parent[normal] node[pos=.6, above=-3pt, sloped] {\footnotesize$1-2\epsilon$}
							}
							child {
								node[dot, below=.5cm, label={[above=-2pt] {$1$}}] (02) {} node[right=.5cm of 02.center] {\normalsize $\vdots$}
								edge from parent[normal]
							}
							edge from parent[emph]
						}
						edge from parent[normal] node[pos=.5, above, sloped, above=-2pt] {\footnotesize$\epsilon$}
					}
					child {
						node[dot, label={[above=0pt] {$0$}}] () {} node[right=.25cm] {$\cdots$}
						edge from parent[normal] node[pos=.6, above=-3pt, sloped] {\footnotesize$1-2\epsilon$}
					}
					child {
						node[dot, above=0cm, label={[above=0pt] {$1$}}] (z1top) {}
						child {
							node[dot, label={[above=0pt] {$0$}}] () {}
							child {
								node[dot, label={[above=-3pt] {$-1$}}] () {} node[right=.25cm] {\normalsize $(2,-1)$}
								edge from parent[normal]
							}
							child {
								node[dot, label={[above=-2pt] {$0$}}] () {} node[right=.25cm] {\normalsize $(2,0)$}
								edge from parent[normal]
							}
							child {
								node[dot, label={[above=-2pt] {$1$}}] () {} node[right=.25cm] {\normalsize $(2,1)$}
								edge from parent[normal]
							}
							edge from parent[emph, dashed] 
						}
						child {
							node[dot, label={[above=0pt] {$1$}}] (x2top) {}
							child {
								node[dot, label={[above=-3pt] {$-1$}}] () {} node[right=.25cm] {\normalsize $(2,0)$}
								edge from parent[normal]
							}
							child {
								node[dot, label={[above=-2pt] {$0$}}] () {} node[right=.25cm] {\normalsize $(2,1)$}
								edge from parent[normal]
							}
							child {
								node[dot, label={[above=-2pt] {$1$}}] (ylabel) {} node[right=.25cm] {\normalsize $(2,2)$}
								edge from parent[normal]
							}
							edge from parent[emph] node[pos=.6, above=-2pt, sloped] {\footnotesize$\beta=1$}
						}						
						edge from parent[normal] node[pos=.6, above, sloped] {\footnotesize$\epsilon$}
					}
					edge from parent[emph] node[pos=.6, above, sloped] {\footnotesize$\alpha=1$}
				}
				edge from parent[normal] node[pos=.4, above=3pt] {$\frac{1}{2}$}
			};
			\node [above=1cm of v1.center] (X1) {\normalsize $X$};
			\node [above=1cm of x1top.center] (Z1) {\normalsize $Z$};
			\node [above=1cm of z1top.center] (X2) {\normalsize $X$};
			\node [above=1cm of x2top.center] (Z2) {\normalsize $Z$};
			\node [below=.5cm of v0.center] (vs) {};
			\node [below=.25cm of x1bottom.center] (x1bottom+) {};
			\node [below=0cm of z1bottom.center] (z1bottom+) {};
			\node [below=.5cm of x2bottom.center] (x2bottom+) {};
			\node [above right=.4cm and .25cm of ylabel.center] () {\normalsize $(y_1, y_2)$};
			\node [right=.25cm of cdots.center] {$\cdots$};
			\node [right=.25cm of cdots1.center] {$\cdots$};
			\node () [rectangle, thick, rounded corners, draw, loosely dotted, fit=(X1) (vs), inner xsep=0pt, inner ysep=8pt] {};
			\node () [rectangle, thick, rounded corners, draw, loosely dotted, fit=(Z1) (x1bottom+), inner xsep=0pt, inner ysep=8pt] {};
			\node () [rectangle, thick, rounded corners, draw, loosely dotted, fit=(X2) (z1bottom+), inner xsep=0pt, inner ysep=8pt] {};
			\node () [rectangle, thick, rounded corners, draw, loosely dotted, fit=(Z2) (x2bottom+), inner xsep=0pt, inner ysep=8pt] {};
		\end{tikzpicture}
		\caption{The pure strategy $\xi$ (solid lines) where the insider buys when the true value is high ($v=1$) and does not trade when it is low ($v=0$).}
	\end{subfigure}\hfill
	\begin{subfigure}[b]{0.33\textwidth} \centering
		\tikzstyle{dot} = [draw, solid, fill, circle, inner sep=0.0cm, outer sep=0cm, minimum size=5pt]
		\tikzstyle{every node}=[font=\footnotesize]
		\begin{tikzpicture}[xscale=3,yscale=5.5]				
			\draw[->, thick] (-.02,0) -- (1.1,0) node[right] {\normalsize $\alpha$};
			\draw[->, thick] (0,-.02) -- (0,.75) node[above left=0cm and -1.3cm, align=left] {\normalsize profits\\\normalsize after $v=1$};
			
			\foreach \pos in {0,.2,.4,.6,.8,1}
			\draw[shift={(\pos,0)}] (0,.02) -- (0,-.02) node[below] {$\pos$};
			
			\foreach \pos in {0,.1,.2,.3,.4,.5,.6,.7}
			\draw[shift={(0,\pos)}] (.02,0) -- (-.02,0) node[left] {$\pos$};
			
			\draw[thick] plot[smooth] file {plot1.table};
			\draw[thick] plot[smooth] file {plot2.table};
			
			\draw (1,.32) node[below] {$x_1=1$};
			\draw (1,.5) node[above] {$x_1=0$};
			
			\draw[dashed] (0.7746420901,0.3350563687) -- (0.7746420901,-.15) node[below, below right= -0.1cm and -1.1cm] {\normalsize $\alpha^\star=0.77464\ldots$}; 
		\end{tikzpicture}
		\caption{The profits after $v=1$ for the pure strategies $x_1=1$ and $x_1=0$ as a function of the buying probability $\alpha$ assumed by the market maker (here we set $\eps=1/8$).} 
	\end{subfigure}
	\caption[]{Visualisation for the game from Example~\ref{example:no_pure_equilibrium} that has no equilibrium with pure strategies.} 
	\label{figure:no_pure_equilibrium}
\end{figure}%

\section{Structure of equilibria in single-period models} \label{section:single_period_structure}

In this section, we establish basic properties of equilibria in single-period Kyle games, i.e., $T=1$ and $\cI_1=\{\{1\},\{2\},\ldots,\{N\}\}$, and show that the insider's demands are uniformly bounded 
if the noise trader's demands lie in $[-1,1]$ and the probabilities of $\{-1\},\{1\}$ are bounded away from zero. 
It also becomes apparent that in discrete state Kyle games the selection of $E_Z$ and $E_X$ is a crucial issue.
We denote $E_V:=\{v^N,v^{N-1},\ldots,v^1\}$ and write an insider strategy as $\xi(v,\{x\})$, $v\in E_V$, $x\in E_X$.
Throughout the section, we assume that
\begin{equation*}
	\{-1,1\}\subseteq E_Z\subseteq [-1,1],\qquad\{0\}\in E_X,
\end{equation*}
and use the notation 
\begin{equation*}
	E^\xi_X:=\{x\in E_X : \exists v\in E_V\ \mbox{with}\ \xi(v,\{x\})>0\}.
\end{equation*}

\begin{assumption}\label{24.8.2023.2}
	One has $((E_X+E_Z)-E_X)\cap {\rm conv}(E_Z)\subseteq E_Z$. In other words, for every\
	$x_1,x_2\in E_X$, $z_1\in E_Z$,
	\begin{equation*}
	x_2\in[x_1+z_1-1,x_1+z_1+1]\quad\implies\quad \exists z_2\in E_Z\ \mbox{such that}\ \
	x_2+z_2 = x_1+z_1.
	\end{equation*}
\end{assumption}

The assumption is satisfied, for instance, if $E_Z$ is an equidistant grid and distances between elements 
of $E_X$ are multiples of the size of this grid. It rules out the effect that by the precise knowledge of $x+z$ 
the market maker can infer $x$, although $x+z$ is no extreme point of $E_X+E_Z$. 

\begin{lemma}\label{14.8.2023.1}
	Let $(\xi,S)$ be an equilibrium in a discrete single-period Kyle game. We have that 
	\begin{itemize}
		\item[(i)] $\xi(v^i,\{x\})>0$, $v^i<v^j$\ $\implies$\ $\xi(v^j,E_X\cap[x,\infty))=1$\quad for all $v^i,v^j\in E_V, x\in E_X$\\
		$\xi(v^i,\{x\})>0$, $v^i>v^j$\ $\implies$\ $\xi(v^j,E_X\cap(-\infty,x])=1$\quad for all $v^i,v^j\in E_V, x\in E_X$
		\item[(ii)] $x>0$, $\xi(v,E_X\cap[x-2,x))=0$ for all $v\in E_V$ 
		$\implies$\ $\int S(x'+z)\zeta(dz)=v'$\\ for all $x\in E_X$, $x'\in E_X\cap[x,\infty)$, $v'\in E_V$ with $\xi(v',\{x'\})>0$\\
		(the inverse assertion for $x<0$ holds as well)
		\item[(iii)] For $y_1,y_2\in E_Y$ with $p^\xi_Y(y_1),p^\xi_Y(y_2)>0$ and $y_2\ge y_1+2$, we have $S(y_1)\le S(y_2)$.
		\item[(iv)] For $y_1,y_2\in E_Y$ with $p^\xi_Y(y_1),p^\xi_Y(y_2)>0$, $y_2\ge y_1$, and $S(y_1)=v^1$, we have $S(y_2)=v^1$.\\
		For $y_1,y_2\in E_Y$ with $p^\xi_Y(y_1),p^\xi_Y(y_2)>0$, $y_2\ge y_1$ and $S(y_2)=v^N$, we have $S(y_1)=v^N$.
	\end{itemize}
\end{lemma}
Property~(i) says that the insider's demand is nondecreasing in the true value she observes. Property~(ii) states that the gaps between the
insider's order sizes (that depend on the true value) should not be larger than the range of the noise trader's order sizes. 
Otherwise, the market maker could infer the true value and profits vanish.  
Maybe surprisingly, the price function~$S$ is in general not nondecreasing, see (Counter-)Example~\ref{example:S_not_monotone}, but we have the weaker properties~(iii) and (iv).  
\begin{proof}[Proof of Lemma~\ref{14.8.2023.1}]
	Ad (i). By symmetry, we only have to prove the first implication. Consider the gain function 
	\begin{equation}\label{1.9.2023.1}
	x\mapsto x(v-\int S(x+z)\zeta(dz)).
	\end{equation}
	If $x$ is a maximiser (not necessarily unique) for $v=v^i$, it strictly dominates $x'< x$ for $v=v^j$.
	
	Ad (ii). Let $x_1\in E_X\cap(0,\infty)$ with $\xi(v,E_X\cap[x_1-2,x_1))=0$ for all $v\in E_V$. Define $x_2:=\inf\{x\in E^\xi_X\cap[x_1,\infty)\}$, $v^\star:=\inf\{v\in E_V : \xi(v,\{x_2\})>0\}$, and $x_3:=\sup\{x\in E_X : \xi(v^\star,\{x\})>0\}$. 
	Here, the case $x_2=\infty$ is trivial and thus excluded.
	By optimality of $\xi$, $0\in E_X$, and part~(i), we have that 
	\begin{equation}\label{24.8.2023.1}
	\int S(x+z)\zeta(dz)\le v^\star\le v\ \mbox{for all}\ x\in E_X\cap[x_1,x_3], v\in E_V\ \mbox{with}\ \xi(v,\{x\})>0.
	\end{equation}
	On the other hand, for each $x\in E^\xi_X\cap[x_1,x_3]$, $z\in E_Z$, the price $S(x+z)$ lies in the convex hull of $\{v\in E_V : \exists x'\in E_X\ \mbox{with}\ \xi(v,\{x'\})>0, z'\in E_Z\ \mbox{such that}\ x'+z' = x+z\}$ by rational pricing (cf. Definition~\ref{definition:rational_pricing_discrete} and observe that the set is nonempty by $x\in E^\xi_X$).
	Consequently, for each $x\in E^\xi_X\cap[x_1,x_3]$, the average price~$\int S(x+z)\zeta(dz)$ lies in the convex hull of
	$M_x:=\{v\in E_V : \exists x'\in E_X\ \mbox{with}\ \xi(v,\{x'\})>0\ \exists z,z'\in E_Z\ \mbox{such that}\ x'+z'=x+z\}$. But for $v\in M_x$ we must have that $v\ge v^\star$ by $x'\ge x-2  \ge x_1-2$ and again part~(i). Together with (\ref{24.8.2023.1}), we obtain that $v=v^\star$ for all
	$x\in E^\xi_X\cap[x_1,x_3]$ and $v\in M_x$. This implies equality in
	(\ref{24.8.2023.1}) and even more that
	\begin{equation}\label{23.8.2023.1}
	S(x+z) = v^\star = v\ \mbox{for all}\ x\in E_X\cap[x_1,x_3], v\in E_V\ \mbox{with}\ \xi(v,\{x\})>0, z\in E_Z.
	\end{equation}
	Now, define $x_4:=\inf\{x\in E_X : \xi(v,\{x\})>0\ \mbox{for some\ }v\in E_V\cap(v^\star,\infty)\}$. We want to show that $x_4>x_3+2$ and assume by contradiction that this does not hold. Then, by $x_4\ge x_3$ and Assumption~\ref{24.8.2023.2}, there would exist a $z'\in E_Z$ such that 
	$x_4+z' = x_3+1$ which would imply that $S(x_3+1)>v^\star$, a contradiction to (\ref{23.8.2023.1}).
	By definition of $x_3$ and $x_4$, the estimate~$x_4>x_3+2$ yields $\xi(v,E_X\cap[x_4-2,x_4))=0$ for all $v\in E_V$. This means that $x_4>x_1$ satisfies the properties that we required for $x_1$. By proceeding analogously, we obtain the assertion because of (\ref{23.8.2023.1}).
	
	Ad (iii). Let $x_1:=\sup\{x\in E^\xi_X : \exists z\in E_Z\ \mbox{such that}\ x+z=y_1\}\le y_1+1\le y_2-1$. Part~(i) and basic properties of the conditional expectation yield that 
	\begin{equation*}
	S(y_1) \le \frac{\sum_i \nu(v^i)\xi(v^i,\{x_1\})v^i}{\sum_i \nu(v^i)\xi(v^i,\{x_1\})}\le S(y_2). 
	\end{equation*}
	Ad (iv). By symmetry, we only have to prove the first implication. We assume by contradiction that there exist $x_2\in E^\xi_X$ with $\xi(v^1,\{x_2\})<1$ and $z_2\in E_Z$ such that $x_2+z_2=y_2$. By $p^\xi_Y(y_1)>0$, there must exist $x_1\in E^\xi_X$ and $z_1\in E_Z$ such that $x_1+z_1=y_1$.
	
	{\em Case} $x_1\le x_2$. It follows from part~(i) that $\xi(v^1,\{x_1\})<1$, a contradiction to $S(y_1)=v^1$.\\
	{\em Case} $x_1>x_2$. We have that $x_2\in[y_1-1,y_1+1]$. By Assumption~\ref{24.8.2023.2}, there exists a $z'\in E_Z$ with $x_2+z'=y_1$ 
	that is a contradiction to $S(y_1)=v^1$.
\end{proof}	

\begin{theorem}\label{theorem:single_period_structure}
	In any equilibrium~$(\xi,S)$ of a discrete single-period Kyle game, either insider's buy orders cannot be executed at a price below the maximal true value, i.e., $S(x+z)=v^1$ for all $x\in E_X\cap(0,\infty)$, $z\in E_Z$,  or their sizes are bounded by $6+6/\zeta(\{1\})$, i.e., $E^\xi_X\cap(0,\infty)\subseteq (0,6+6/\zeta(\{1\}))$. Analogously, either insider's sell orders cannot be executed at a price above the minimal true value or their sizes are bounded by $6+6/\zeta(\{-1\})$. 
\end{theorem}
\begin{example}
	Both the true value and the demand of the noise trader are $\pm 1$ with probability $1/2$. In addition, $E_X=\{-2n,-2(n-1),\ldots,0,2,\ldots,2n\}$ for some
	fixed $n\in\bN$. Then, one equilibrium is given by $\xi(v,\cdot)=\delta_{2n} 1_{(v=1)} + \delta_{-2n} 1_{(v=-1)}$ and
	$S(y)=-1_{(y<0)} + 1_{(y>0)}$ for all $y\in E_Y$. 
\end{example}

\begin{proof}[Proof of Theorem~\ref{theorem:single_period_structure}]
	We fix a Kyle game and an equilibrium~$(\xi,S)$ of this game. By symmetry, it is sufficient to show the assertion regarding the insider's buy orders. Let $x_1:=\sup\{x\in E_X\cap(0,\infty) : \xi(v^1,\{x\})>0\ \mbox{and}\ \int S(x+z)\zeta(dz)<v^1\}$. 
	
	{\em Case} $x_1=-\infty$. Let $x\in E^\xi_X\cap(0,\infty)$ (if no such $x$ exists we are done). 
	By Lemma~\ref{14.8.2023.1}(i), we have that $\xi(v^1,E_X\cap[x,\infty))>0$. This means that there is an order size $x'\in E^\xi_X\cap[x,\infty)$ with $\xi(v^1,\{x'\})>0$ that is an optimizer of (\ref{1.9.2023.1}) when the true value takes the maximal value~$v^1$. But, by $x_1=-\infty$, the gain~$v^1-\int S(x'+z)\zeta(dz)$ cannot be positive. Consequently, we must have that $\int S(x''+z)\zeta(dz)=v^1$ for all $x''\in E_X\cap(0,\infty)$
	since otherwise $x'$ could not be optimal given $v^1$. It follows that $S(x''+z)=v^1$ for all $z\in E_Z$. We conclude that the insider's buy orders are never executed at a price below $v^1$. 
	
	{\em Case} $x_1>0$.  Now, we turn to the ``main'' case. W.l.o.g. $x_1\ge 6$ since otherwise $x_1<6+6/\zeta(\{1\})$ and the upper bound could be verified similar to the previous case. Let $v^\star:=\inf\{v\in E_V : \xi(v,x_1)>0\}$. 
	By definition of $x_1$ and $v^\star$, we have that $\int S(x_1+z)\zeta(dz)<v^1$ and $\int S(x_1+z)\zeta(dz)\le v^\star$.  
	Let us find an $x_2\in E_X$ with 
	\begin{equation}\label{28.8.2023.1}
	S(x_2+1)\le\int S(x_1+z)\zeta(dz)\quad\mbox{and}\quad S(x_2+1)<v^\star.
	\end{equation}
	First, consider the case that $v^\star=v^1$. The first inequality in (\ref{28.8.2023.1}) is satisfied by any $x_2\in E^\xi_X\cap(-\infty,x_1-4]$ because of Lemma~\ref{14.8.2023.1}(iii). Lemma~\ref{14.8.2023.1}(ii) 
	guarantees that there exists an $x_2\in E^\xi_X\cap[x_1-6,x_1-4)$ since otherwise the insider could not make a profit on average with $x_1$ when $v^1$ occurs.
	Then, the second inequality in (\ref{28.8.2023.1}) follows from $\int S(x_1+z)\zeta(dz)<v^1$.
	
	We proceed with the case that $v^\star<v^1$. Since $\int S(x_1+z)\zeta(dz)\le v^\star$, $\xi(v^1,\{x_1\})>0$ and by rational pricing (cf. Definition~\ref{definition:rational_pricing_discrete}), there must exist $x\in E^\xi_X$, $z_1,z\in E_Z$, and $v\in E_V\cap(-\infty,v^\star)$  
	with $\xi(v,\{x\})>0$ and $x+z=x_1+z_1$. One has that $x\in E^\xi_X\cap[x_1-2,\infty)$. Consequently, for all $x'\in E^\xi_X\cap(-\infty,x_1-4)$
	we get $S(x'+1)\le v<v^\star$ by Lemma~\ref{14.8.2023.1}(i) and rational pricing.
	Thus, $x_2$ from above can also be taken in the case~$v^\star<v^1$ and (\ref{28.8.2023.1}) is shown. In addition, since the total demand $x_1+1$ can only occur if $v\ge v^\star$, rational pricing yields
	\begin{equation}\label{14.8.2023.3}
	S(x_1+1)\ge v^\star.
	\end{equation}
	Define $s:=S(x_2+1)$ and $\Delta := x_1(v^\star-\int S(x_1+z)\zeta(dz)) - x_2(v^\star-\int S(x_2+z)\zeta(dz))$.
	We have that 
	\begin{align*} 
		\Delta  &= (x_1-x_2)(v^\star-\int S(x_1+z)\zeta(dz)) + x_2( \int S(x_2+z)\zeta(dz) - \int S(x_1+z)\zeta(dz))\\
		&\le  (x_1-x_2)(v^\star-s) +x_2 (s-v^\star)\zeta(\{1\})\\
		&\le  (v^\star-s)(6 - x_2\zeta(\{1\})),
	\end{align*}
	where we use (\ref{14.8.2023.3}) to estimate $S(x_2+z)-S(x_1+z)$ for $z=1$. Since $x_1$ maximises (\ref{1.9.2023.1}) given $v^\star$, we have $\Delta\ge 0$.
	By $v^\star>s$, we obtain that $x_2\le 6/\zeta(\{1\})$ and arrive at $x_1\le 6 + 6/\zeta(\{1\})$. 
\end{proof} 
\begin{remark}
	We note that (\ref{14.8.2023.3}) need not hold for $S(x_1+z)$ 
	with $z<1$ instead of $S(x_1+1)$. This is the reason why the upper bound~$6+6/\zeta(\{1\})$ depends on the probability of the noise trader's demand at the boundary. For a sequence of discrete models, the bound could tend to infinity if $\zeta(\{1\})\to 0$.
\end{remark}
\begin{example}[$\nexists$ nondecreasing equilibrium price function] \label{example:S_not_monotone}
	We provide a minimalist example of a single-period Kyle game in which the equilibrium is unique and the equilibrium price function is strictly decreasing for some demands. The true value takes a high value $v^1:=1$, a medium value $v^2:=\sfrac{1}{2}$, or a low value $v^3 := 0$, with uniform distribution $\nu:=\sfrac{1}{3} \delta_{0} + \sfrac{1}{3} \delta_{\sfrac{1}{2}} + \sfrac{1}{3} \delta_{1}$. Both insider and noise trader can buy or sell one share of the asset or not trade at all, i.e., $E_X := E_Z := \{ -1, 0, 1\}$, and so the market maker observes a total order flow in $E_Y = \{ -2, -1, 0, 1, 2\}$. The crux in our example is that the noise trader shows bearish sentiment and is more likely to sell, namely  
	\begin{equation*}
		\zeta := \frac{6}{8} \delta_{-1} + \frac{1}{8} \delta_{0} + \frac{1}{8} \delta_1.
	\end{equation*}
	This leads to the effect that when the market maker observes $y=0$ he is more likely to believe that the insider and noise trader traded $x=1$ and $z=-1$, respectively, than he would be, ceteris paribus, under uniformly distributed noise trades. The conditional probability of $x=1$ can be higher under the condition $y=0$ than under $y=1$. This is the reason why the example can work although the insider's demand is nondecreasing in the true 
	value (in the sense of Lemma~\ref{14.8.2023.1}(i)) and the market maker's price is the conditional expectation of the true value given $y$. 
	
	In the following, we want to show that the Kyle equilibrium is unique, and its price function~$S$ satisfies $S(0)>S(1)$. For this, let $(\xi, S)$ be an arbitrary equilibrium.

	\emph{Step 1:} Let us first show that $\xi(1, \,\cdot\,)=\delta_1$ and $\xi(0, \,\cdot\,)=\delta_{-1}$, i.e., for the extreme true values 
	buying/selling is the unique optimal action for the insider. 
	Assume by contradiction that $\xi(1, \,\cdot\,)\not=\delta_1$, i.e., buying is at least not the only optimal action when $v=1$. Because $S\le 1$ 
	and $E_X = \{ -1, 0, 1\}$,
	this can only be the case if the profits are zero, i.e., the average price~$\int S(1+z)\zeta(dz)$ equals $1$, and this is equivalent 
	to $S(2)=S(1)=S(0)=1$. 
	But regardless of $\xi(\sfrac{1}{2}, \,\cdot\,)$ and $\xi(0, \,\cdot\,)$, the noise trader always ensures that $y=0$ is reached with positive probability when $v = \sfrac{1}{2}$ or $v = 0$, and so $S(0)<1$ by rational pricing, a contradiction. 
	Hence, we must have that $\xi(1, \,\cdot\,) = \delta_1$. By the same arguments, $\xi(0, \,\cdot\,) = \delta_{-1}$ is the unique optimal action when $v=0$.
	
	\emph{Step 2:} Now we turn to the optimal strategy when $v=\sfrac{1}{2}$. Using the notation $\xi(\sfrac{1}{2}, \,\cdot\,) = \alpha_1 \delta_1 + \alpha_0 \delta_{0} + \alpha_{-1} \delta_{-1}$, with $\alpha_1, \alpha_{0}, \alpha_{-1} \ge 0$ and $\alpha_1+\alpha_{0}+\alpha_{-1}=1$, we can compute the rational prices given $\xi$:
	\begin{align*}
		S(2) & = \frac{1 \cdot 1 + \alpha_1 \cdot \frac{1}{2}}{1 + \alpha_1}= \frac{1  + \frac{1}{2}\alpha_1 }{1 + \alpha_1},\\
		S(1) & = \frac{1 \cdot 1 + (\alpha_1 +\alpha_{0}) \cdot \frac{1}{2}}{1 + \alpha_1 + \alpha_{0}} = \frac{1 + \frac{1}{2}(\alpha_1 +\alpha_{0})}{1 + \alpha_1 + \alpha_{0}} , \\
		S(0) & =  \frac{6 \cdot 1 + (6\alpha_1 +\alpha_{0}+\alpha_{-1}) \cdot \frac{1}{2} + 1 \cdot 0}{6 + 6\alpha_1 + \alpha_{0} + \alpha_{-1} + 1} 
		= \frac{6  + \frac{1}{2} (6\alpha_1 +\alpha_{0}+\alpha_{-1}) }{7 + 6\alpha_1 + \alpha_{0} + \alpha_{-1} },	\\
		S(-1) & = \frac{(6\alpha_{0}+\alpha_{-1}) \cdot \frac{1}{2} + 1 \cdot 0}{6\alpha_{0} + \alpha_{-1}+1} = \frac{\frac{1}{2}(6\alpha_{0}+\alpha_{-1})}{1+ 6\alpha_{0} + \alpha_{-1}},  \\
		S(-2) & = \frac{6\alpha_{-1} \cdot \frac{1}{2} + 6 \cdot 0}{6\alpha_{-1} + 6} = \frac{\alpha_{-1}}{2+2\alpha_{-1}}. 	
	\end{align*}
	We observe that $S(2), S(1), S(0) > \sfrac{1}{2}$ and so the average price
	\begin{equation*}
		\int S(1+z)\zeta(dz) > \frac{1}{2},
	\end{equation*}
	which makes buying in $v=\sfrac{1}{2}$ suboptimal. As a result $\alpha_1=0$.	
	In addition, the above formulas show immediately that $S(0) < 1$, $S(-1) < \sfrac{1}{2}$, and $S(-2) \le \sfrac{1}{4}$. Thus, we get the rough estimate 
	\begin{equation*}
		\int S(-1+z)\zeta(dz) = \frac{1}{8} S(0) + \frac{1}{8} S(-1) + \frac{6}{8} S(-2) < \frac{6}{16}  < \frac{1}{2}.
	\end{equation*}
	Consequently, selling in $v=\sfrac{1}{2}$ is not profitable either, so $\alpha_{-1}=0$, and $\alpha_{0}=1$. In other words, the optimal strategy is $\xi(\sfrac{1}{2}, \,\cdot\,) = \delta_{0}$. In conclusion, the unique Kyle equilibrium $(\xi, S)$ is given by
	\begin{equation*}
		\xi(v, \,\cdot\,) := \delta_{2v-1}, \qquad S(-2)=0,\ S(-1)=\frac{3}{7},\ S(0)=\frac{13}{16},\ S(1)=\frac34,\ S(2)=1,
	\end{equation*}
	also depicted in Figure \ref{figure:example_S_not_monotone}. The price function $S$ is decreasing between $0$ and $1$.
\end{example}%
\noindent\begin{figure}[]%
	\tikzstyle{dot} = [draw, solid, fill, circle, inner sep=0.0cm, outer sep=0cm, minimum size=5pt]
	\tikzstyle{densely dashed}=[dash pattern=on 6pt off 1.5pt]
	\tikzstyle{every node}=[font=\footnotesize]
	\tikzstyle{normal} = [thin, solid]
	\tikzstyle{emph} = [very thick, solid]
	\tikzstyle{level 1}=[level distance=1.25cm, sibling distance=2.3cm]
	\tikzstyle{level 2}=[level distance=1.5cm, sibling distance=.65cm]
	\tikzstyle{level 3}=[level distance=1.25cm, sibling distance=.65cm]
	\begin{subfigure}[b]{0.5\textwidth} \centering
		\begin{tikzpicture}[grow=right]
			\node[dot] (r) {}
			child {
				node[dot, above=.75cm, label={[below=5pt] {$v=0$}}] (v0) {}
				child {
					node[dot, below=.5cm, label={[above=0pt] {$-1$}}] (x0) {}
					child {
						node[dot, label={[right=.5cm] {\normalsize $-2$}}] () {} node[above] {$-1$}
						edge from parent[normal]
					}
					child {
						node[dot, label={[right=.5cm] {\normalsize $-1$}}] () {} node[above=2pt] {$0$}
						edge from parent[normal]
					}
					child {
						node[dot, label={[right=.5cm] {\normalsize $0$}}] () {} node[above=2pt] {$1$}
						edge from parent[normal]
					}
					edge from parent[emph] 
				}
			}
			child {
				node[dot, label={[above=-1pt] {$v=\frac{1}{2}$}}] (v12) {}
				child {
					node[dot, label={[above=0pt] {$0$}}] (2) {}        
					child {
						node[dot, label={[right=.5cm] {\normalsize $1$}}] {} node[above] {$-1$}
						edge from parent[normal] 
					}
					child {
						node[dot, label={[right=.5cm] {\normalsize $0$}}] {} node[above=2pt] {$0$}
						edge from parent[normal] 
					}
					child {
						node[dot, label={[right=.5cm] {\normalsize $1$}}] {} node[above=2pt] {$1$}
						edge from parent[normal] 
					}
					edge from parent[emph]
				}
			}
			child {
				node[dot, below=.75cm, label={[above=0pt] {$v=1$}}] (v1) {}
				child {
					node[dot, above=.5cm, label={[above=0pt] {$1$}}] (x1) {}
					child {
						node[dot, label={[right=.5cm] {\normalsize $0$}}] () {} node[above] {$-1$}
						edge from parent[normal]
					}
					child {
						node[dot, label={[right=.5cm] {\normalsize $1$}}] () {} node[above=2pt] {$0$}
						edge from parent[normal]
					}
					child {
						node[dot, label={[right=.5cm] {\normalsize $2$}}] (yt) {} node[above=2pt] {$1$}
						edge from parent[normal]
					}
					edge from parent[emph]
				}
			};
			\node [above=1.25cm of v1.center] (X) {\normalsize $X$};
			\node [above=1cm of x1.center] (Z) {\normalsize $Z$};
			\node [below=.5cm of v0.center] (vs) {};
			\node [below=.25cm of x0.center] (xs) {};
			\node [above right=.5cm and .5cm of yt] () {\normalsize $y$};
			\node () [rectangle, thick, rounded corners, draw, loosely dotted, fit=(X) (vs), inner xsep=5pt, inner ysep=8pt] {};
			\node () [rectangle, thick, rounded corners, draw, loosely dotted, fit=(Z) (xs), inner xsep=5pt, inner ysep=8pt] {};
		\end{tikzpicture}
		\caption{The optimal (non-randomising) strategy of the insider in the game tree.}
	\end{subfigure}
	\begin{subfigure}[b]{0.5\textwidth} \centering
		\tikzstyle{dot} = [draw, solid, fill, circle, inner sep=0.0cm, outer sep=0cm, minimum size=5pt]
		\tikzstyle{every node}=[font=\footnotesize]
		\begin{tikzpicture}				
			\draw[->, thick] (-3,-2.5) -- (2.5,-2.5) node[right] {\normalsize $y$};
			\draw[->, thick] (-2.5,-3) -- (-2.5,2.5) node[above] {\normalsize $S(y)$};
			
			\foreach \pos in {-2,-1,0,1,2}
			\draw[shift={(\pos,-2.5)}] (0pt,2pt) -- (0pt,-2pt) node[below] {$\pos$};
			
			\draw[shift={(-2.5,-2)}] (2pt,0pt) -- (-2pt,0pt) node[left] {$0$};
			\draw[shift={(-2.5,0)}] (2pt,0pt) -- (-2pt,0pt) node[left] {$\frac{1}{2}$};	
			\draw[shift={(-2.5,2)}] (2pt,0pt) -- (-2pt,0pt) node[left] {$1$};	
			
			\path (-2,-2) node[dot] (-2) {} 
			(-1,-2/7) node[dot] (-1) {} 
			(0,10/8) node[dot] (0) {} 
			(1,2/2) node[dot] (+1) {} 
			(2,2) node[dot] (+2) {};
			\draw[dashed] (-2.center) -- (-1.center) -- (0.center) -- (+1.center) -- (+2.center);
		\end{tikzpicture}
		\caption{The rational price function $S$ is decreasing between $y=0$ and $y=1$.}
	\end{subfigure}
	\caption[]{The unique equilibrium from Example~\ref{example:S_not_monotone}.} 
	\label{figure:example_S_not_monotone}
\end{figure}
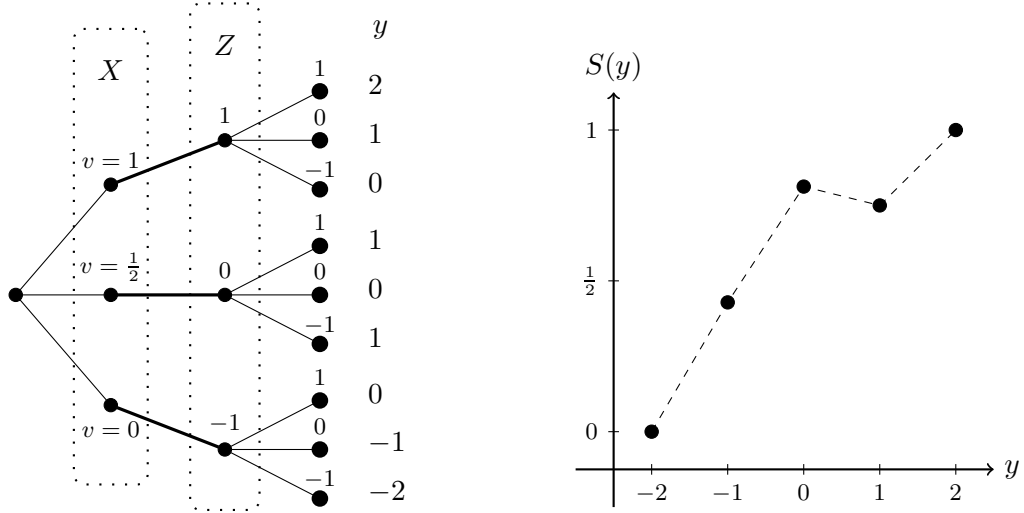%
\section{Continuous state game} \label{section:continuous_state_game}

In this section, we prove the existence of an equilibrium in the single-period Kyle game when the true value and the noise trader's demand are probability measures over the Borel $\sigma$-algebra~$\cB(\bR)$ with compact but not necessarily finite support. The main challenge is that the bounded price function of the market maker in the continuous game can only be expected to be measurable but in general not a continuous function of the total demand. Consequently, the insider's profit need not be continuous in her demand. The set of $[-1,1]$-valued Borel-measurable functions equipped with pointwise convergence almost everywhere is not sequentially compact. The standard example is a sequence of Rademacher functions (see, e.g., Example~3.2 in \citep{balder_new_2021}) that cannot have a convergent subsequence. Because of the lack of compactness of the set of price functions, we cannot apply standard infinite-dimensional fixed point theorems of Schauder-Tychonoff's or Kakutani-Fan's type (see, e.g., Theorem~10.1 and Theorem~13.1 in \cite{pata_fixed_2019}, respectively) directly to the continuous game. 

Instead, we consider equilibria of a sequence of finite Kyle games that are derived from the continuous game by discretising the true value and the noise trader's demand.
We assume that the noise trader's demand is absolutely continuous w.r.t. the Lebesgue measure~$\lambda$. This tames the expected gains of the insider, and a weak limit of her discrete equilibrium strategies (that exists along a subsequence by compactness of probability distributions on $[-1,1]$) is part of an equilibrium in the continuous model. On the other hand, after passing to forward convex combinations, the discrete equilibrium price functions of the market maker possess a pointwise limit $\lambda$-almost everywhere. Together with the insider strategy, it is an equilibrium price function in the continuous model.

\smallskip

Let us introduce the continuous state model. Let $E_V := [0,1]$, $\nu \in \cP(E_V, \cB(E_V))$, $E_Z := [-1,1]$, and $\zeta \in \cP(E_Z, \cB(E_Z))$. We assume that
\begin{equation*}
	\zeta\ \text{is absolutely continuous w.r.t. the Lebesgue measure } \lambda, \text{ i.e.,\ }\zeta \ll \lambda.
\end{equation*}%

\begin{definition}[Insider strategy] \label{definition:young_measure}
	Let $E_X=[\underline{x},\overline{x}]$ for some fixed $\underline{x},\overline{x}\in\bZ$ with $\underline{x}<0<\overline{x}$. Based on $(E_V, \cB(E_V), \nu)$, an insider strategy is a \textit{Young measure} $\xi: E_V \times \cB(E_X) \to [0,1]$, i.e., 
	\begin{enumerate}[(i)]
		\item $\xi(\;\cdot\;, B)$ is Borel-measurable for every $B \in \cB(E_X)$,
		\item $\xi(v, \;\cdot\;)$ is a probability measure for every $v \in E_V$.
	\end{enumerate}
	With $\Xi:=\Xi(E_V, \nu; E_X)$ we denote the set of insider strategies.
\end{definition}

The interpretation is analogous to that of discrete behaviour strategies (see Definition \ref{definition:behaviour_strategy}). We refer to \cite{balder_new_2021} and the references therein for an overview of the theory of Young measures and their applications in 
optimal control theory. The market maker observes the total order flow $y := x + z \in E_Y := E_X+E_Z = [\underline{x}-1,\overline{x}+1]$ and sets a price.

\begin{definition}[Price function of market maker]
	A price function is a Borel-measurable function
	\begin{equation*} 
		S : E_Y  \to E_V, \qquad
		y  \mapsto S(y).
	\end{equation*}
	With $\cS$ we denote the set of price functions.
\end{definition}

Given a price function $S \in \mathcal{S}$, the insider's objective is to maximise the \emph{expected utility}
\begin{equation} \label{equation:expected_utility}
	u(\xi, S) := \iiint \left[v - S(x+z)\right] x \;\zeta(dz) \xi(v, dx) \nu(dv) \rightarrow \max_{\xi \in \Xi} !
\end{equation}
Analogous to the discrete Kyle game, a price function~$S$ is \emph{rational} assuming a strategy $\xi \in \Xi$ of the insider if for all $A \in \cB(E_Y)$:
\begin{equation} \label{equation:rational_pricing_continuous}
	\iiint S(x+z) \Ind{A}(x+z) \;\zeta(dz)\xi(v, dx)\nu(dv) = \iiint v \;\Ind{A}(x+z) \; \zeta(dz)\xi(v, dx)\nu(dv).
\end{equation}
An equilibrium in the continuous state Kyle game has the familiar structure of simultaneously requiring optimality for the insider, and rational pricing for the market maker:

\begin{definition}[Continuous Kyle equilibrium]
	A \emph{continuous Kyle equilibrium} is a pair $(\xi_\star, S_\star) \in \Xi \times \mathcal{S}$ satisfying 
	\begin{enumerate}[(i)]
		\item Profit maximisation: Given $S_\star$, the strategy $\xi_\star$ maximises (\ref{equation:expected_utility}).
		\item Rational pricing: Given $\xi_\star$, the price system~$S_\star$ is rational according to (\ref{equation:rational_pricing_continuous}),
	\end{enumerate} 
\end{definition}

The main result of this section is the following theorem, which we prove in the remainder of the section.

\begin{theorem} \label{theorem:continuous_trade}
	The single-period, continuous state Kyle game admits a Kyle equilibrium. 
\end{theorem}

\subsection{Discretisation and embedding} \label{subsection:results_proofs}

We construct a sequence of discrete state games acting on the refining sequence of \emph{dyadic, equidistant grids}
\begin{equation*}
	\begin{gathered}
	E_V^n := \left\{ \frac{k}{2^{n}}: k = 0,1,\ldots, 2^n\right\},\quad 
	E_Z^n := \left\{ \frac{k}{2^n}: k = -2^n, \ldots, 0, \ldots, 2^n\right\},\\
	\text{and}\quad E_X^n := \left\{ \frac{k}{2^n}: k = 2^n \underline{x}, \ldots, 0, \ldots, 2^n \overline{x}\right\},
	\quad n\in\bN.
	\end{gathered}
\end{equation*}
We embed the discrete strategies and price functions into the continuous model through piecewise constant continuation between the points of the $n$\textsuperscript{th} grid:
\begin{equation}\label{equation:embedding}
	\begin{aligned}
		\Xi^n &:= \left\{ \xi \in \Xi :\; \xi(v, \,\cdot\,) = \xi\left(\frac{\floor*{ 2^n v } }{2^n}, \;\cdot\; \right) \text{ and } \xi(v, E_X \setminus E^n_X)=0 \text{ for all } v\in E_V \right\},\\
		\cS^n &:= \left\{ S \in \cS:\; S(y) = S\left( \frac{\floor*{ 2^n y } }{2^n}\right) \text{ for all } y\in E_Y\right\}.
	\end{aligned}
\end{equation}
It is straightforward to check that the insider's expected utility from (\ref{equation:expected_utility_discrete}) in the $n$\textsuperscript{th} discrete game that is created by discretising the measures $\nu$ and $\zeta$, can be expressed in terms of the continuous quantities by
\begin{equation} \label{equation:expected_utility_discrete_nth_game}
	u^n(\xi, S) = \iiint x\left[\frac{\floor*{2^n v}}{2^{n}}-S\left(x+z\right) \right]\, \zeta(dz) \xi(v, dx) \nu(dv) \qquad \text{for } (\xi,S) \in \Xi^n \times \cS^n.
\end{equation}
Analogously, the rational pricing condition from Definition~\ref{definition:rational_pricing_discrete} for $S\in\cS^n$ given $\xi\in\Xi^n$ reads 
\begin{equation}\label{equation:rational_pricing_discrete_nth_game}
	\begin{split}
		&\iiint S(x+z) \Ind{A}\left(x+z\right)\, \zeta(dz) \xi(v, dx) \nu(dv) \\
		&\qquad=  \iiint \frac{\floor*{2^{n} v}}{2^{n}} \Ind{A}\left(x+z\right)\, \zeta(dz) \xi(v, dx) \nu(dv)
	\end{split}
\end{equation}
for all sets $A$ of the form $A=[k/2^n,(k+1)/2^n)$,  $k\in\{(\underline{x}-1)2^n,(\underline{x}-1)2^n+1,\ldots,(\overline{x}+1)2^n-1\}$. 

\subsection{Approximation and existence of a limit point}

For the convenience of the reader, we repeat the following definition.

\begin{definition}[Narrow convergence] \label{definition:narrow_convergence}
	A sequence of Young measures~$(\xi^n)_{n\in\bN}\subseteq\Xi(E_V, \nu; E_X)$ converges \emph{narrowly} to $\xi\in\Xi(E_V, \nu; E_X)$ if for all $A \in \cB(E_V), f \in \cC_b(E_X)$
	\begin{equation*}
		\iint \Ind{A}(v) f(x)\; \xi^n(v, dx)\nu(dv) \to \iint {\Ind{A}(v)} f(x)\; \xi(v, dx)\nu(dv).
	\end{equation*}
\end{definition}

An immediate consequence of narrow convergence is that
\begin{gather} \label{equation:weak_convergence_product_measure}
	\nu \otimes \xi^n\to \nu \otimes \xi\quad\text{weakly, where} \\
	\nu\otimes \xi^{n}(A \times B) := \int_A \xi^{n}(v, B)\; \nu(dv)\ \text{for } A \times B \in \cB(E_V) \times \cB(E_X) \nonumber
\end{gather}
(to see this, one applies Fubini's theorem for transition probabilities).
For further information regarding narrow convergence we refer to \cite{balder_new_2021} and the references therein. 

\begin{lemma}[Convergence of utilities] \label{lemma:convergence_utilities}
	Let $(\xi^n, S^n)_{n\in\bN}\subseteq \Xi^n \times \cS^n$ (that are not necessarily equilibria) such that $\xi^n \to \xi$ narrowly for some $\xi \in \Xi$ and $\sum_{k=0}^{k_n} \lambda_{n,k} S^{n+k} \to S$ $\lambda$-a.e. for some $S \in \cS$ and forward convex combinations~$(\lambda_{n,k})_{n\in\bN, k=0,1,\ldots,k_n}\subseteq\bR_+$, $k_n\in\bN$ with $\sum_{k=0}^{k_n}\lambda_{n,k}=1$. Then
	\begin{equation*}
		\sum_{k=0}^{k_n} \lambda_{n,k} u^{n+k}(\xi^{n+k}, S^{n+k}) \xrightarrow{n \to \infty} u(\xi, S).
	\end{equation*}
\end{lemma}

\begin{proof}
	We split the integrals in (\ref{equation:expected_utility_discrete_nth_game}) and (\ref{equation:expected_utility}) into two parts and handle them  separately. We start with the first parts including the true value. Let $\epsilon >0$. For every $m\in\bN$, we get 
	\begin{equation}\label{equation:integral_approximation_11}
		\begin{aligned}
			\iint \abs*{x \frac{\floor*{2^m v}}{2^m} - xv }  \;  \xi^m(v, dx) \nu(dv) 
			&\le \max(|\underline{x}|,|\overline{x}|)\sup_{v \in [0,1]} \abs*{ \frac{\floor*{2^m v}}{2^m} -v} \\
			&\le \max(|\underline{x}|,|\overline{x}|)\frac{1}{2^m}.
		\end{aligned}
	\end{equation}
	By \cite[Theorem~2.2(b)]{balder_generalized_1988}, $\xi^n \to \xi$ narrowly implies that there exists $n\in\bN$ such that
	\begin{equation} \label{equation:integral_approximation_12}
		\abs*{ \iint xv \;\xi^{n+k}(v, dx) \nu(dv) - \iint xv \;\xi(v, dx) \nu(dv)  } < \epsilon\quad\mbox{for all}\ k\in\bN_0.
	\end{equation}
	Combining the two estimates (\ref{equation:integral_approximation_11}) and (\ref{equation:integral_approximation_12}) yields that for $n$ large enough
	\begin{align}\label{integral_approximation_1} 
		&\abs*{ \sum_{k=0}^{k_n} \lambda_{n,k} \iint x \frac{\floor*{2^{n+k} v}}{2^{n+k}} \;\xi^{n+k}(v, dx) \nu(dv) - \iint xv \;\xi(v, dx) \nu(dv) } \nonumber\\
		&\quad \le \sum_{k=0}^{k_n} \lambda_{n,k} \abs*{  \iint x \frac{\floor*{2^{n+k} v}}{2^{n+k}} \;\xi^{n+k}(v, dx) \nu(dv) - \iint xv \;\xi^{n+k}(v, dx) \nu(dv)  } \nonumber\\
		&\quad\quad + \sum_{k=0}^{k_n} \lambda_{n,k} \abs*{  \iint xv \;\xi^{n+k}(v, dx) \nu(dv)  - \iint xv \;\xi(v, dx) \nu(dv)  } \nonumber\\
		&\quad \le  \sum_{k=0}^{k_n} \lambda_{n,k} \left(\max(|\underline{x}|,|\overline{x}|)\frac{1}{2^{n+k}} + \epsilon \right) < 2\epsilon.
	\end{align}
	Now we turn to the second parts of the integrals in (\ref{equation:expected_utility_discrete_nth_game}) and (\ref{equation:expected_utility})
	including the price functions. This is more challenging as both integrand and integrator vary with $n$. The key is that the absolute continuity of the noise trader's demand~$\zeta$ smooths the realised price $\int S^n(x+z) \zeta(dz)$ per share for an order of size~$x$. Define
	\begin{equation} \label{equation:definition_fn}
		f^{n}(x) := x \int S^{n}(x+z)\;\zeta(dz), \qquad  n\in\bN.
	\end{equation}
	Since $\zeta\ll \lambda$, there exists a density $g\in L^1(\lambda)$ with $\zeta(A) = \int_A g\,d\lambda$ for all $A\in\cB(\bR)$.
	Due to the translation invariance of the Lebesgue measure, we can write
	\begin{equation*}
		f^{n}(x) = x \int S^{n}(x+z) g(z) \,\lambda(dz) = x \int S^{n}(y) g(y-x) \,\lambda(dy).
	\end{equation*}
	Since $g\in L^1(\lambda)$, we have that $\int \abs*{g(z-h) - g(z)} \lambda (dz) \to 0$ as $h \to 0$, see, e.g., \cite[Lemma~2.7]{kallenberg_foundations_2021}. The reason is that $g$ can be approximated in $L^1(\lambda)$ by (uniformly) continuous functions, and for a uniformly continuous function $\wt{g}$ we can use the estimate $\int \abs*{\wt{g}(z-h) - \wt{g}(z)} \lambda (dz) \le \sup_{z\in E_Z} \abs*{\wt{g}(z-h) - \wt{g}(z)} \lambda(E_Z)$ which tends to 0 as $h \to 0$. 
	As a result, we obtain
	\begin{equation}\label{equation:continuity_fn}
		\abs*{f^{n}(x+h) - f^n(x)} \le \sup_{x\in E_X} \abs*{x} \sup_{y\in E_Y} \abs*{S^{n}(y)} \int \abs*{ g(y-x-h) - g(y-x)} \,\lambda(dy) \xrightarrow{h \to 0} 0.
	\end{equation}
	The uniform estimate $\abs*{S^{n}} \le 1$ yields that the family~$(f^n)_{n\in\bN}$ is equicontinuous at every fixed $x$.
	Consequently, we can apply Lemma~\ref{lemma:rao_theorem} and (\ref{equation:weak_convergence_product_measure}) to obtain that	
	\begin{equation*}
		\sup_{m\in\bN} \;\abs*{ \iint f^m(x) \xi^n(v, dx) \nu(dv) - \iint f^m(x) \xi(v, dx) \nu(dv)} \xrightarrow{n \to \infty} 0.
	\end{equation*}
	It follows that for $n$ large enough and all $k\in\bN_0$
	\begin{equation} \label{equation:integral_approximation_21}
		\left| \iint f^{n+k}(x) \xi^{n+k}(v, dx) \nu(dv) - \iint f^{n+k}(x) \xi(v, dx) \nu(dv)  \right| < \epsilon.
	\end{equation}
	Furthermore, by $\zeta\ll \lambda$, $\sum_{k=0}^{k_n} \lambda_{n,k} S^{n+k}\to S$ $\lambda$-a.e., and dominated convergence, we obtain
	\begin{equation*}
		\sum_{k=0}^{k_n} \lambda_{n,k} f^{n+k}(x) = x\int \sum_{k=0}^{k_n} \lambda_{n,k} S^{n+k}(x+z) \; \zeta(dz) \longrightarrow \; x\int S(x+z) \;\zeta(dz) =: f(x) 
	\end{equation*}	
	pointwise. Thus, dominated convergence yields that for $n$ large enough
	\begin{equation} \label{equation:integral_approximation_22}
		\left| \iint \sum_{k=0}^{k_n} \lambda_{n,k} f^{n+k}(x)\; \xi(v, dx) \nu(dv) - \iint f(x)\; \xi(v, dx) \nu(dv)  \right| < \epsilon. 
	\end{equation} 
	Combining the two estimates (\ref{equation:integral_approximation_21}) and (\ref{equation:integral_approximation_22}), we get that for $n$ large enough
	\begin{align*} 
		& \Bigg| \sum_{k=0}^{k_n} \lambda_{n,k} \iiint xS^{n+k}(x+z) \;\zeta(dz)\xi^{n+k}(v, dx) \nu(dv) \\[-2ex]
		&\hspace{16em} - \iiint xS(x+z) \;\zeta(dz)\xi(v, dx)\nu(dv) \Bigg| \\
		&\quad =
		\abs*{ \sum_{k=0}^{k_n} \lambda_{n,k} \iint f^{n+k}(x) \;\xi^{n+k}(v, dx) \nu(dv) - \iint f(x) \;\xi(v, dx)\nu(dv) } \\
		&\quad \le \sum_{k=0}^{k_n} \lambda_{n,k} \abs*{ \iint f^{n+k}(x) \;\xi^{n+k}(v, dx) \nu(dv) - \iint f^{n+k}(x) \;\xi(v, dx)\nu(dv) } \\
		&\qquad +  \abs*{ \iint \sum_{k=0}^{k_n} \lambda_{n,k}f^{n+k}(x) \;\xi(v, dx) \nu(dv) - \iint f(x) \;\xi(v, dx)\nu(dv) } \\
		&\quad < 2\epsilon.
	\end{align*}
	Together with (\ref{integral_approximation_1}), the assertion follows.
\end{proof}

\begin{proof}[Proof of Theorem~\ref{theorem:continuous_trade}]
	
	{\em Step 1:} For each $n\in\bN$, we choose a Kyle equilibrium~$(\xi^n_\star, S^n_\star) \in \Xi^n \times \cS^n\subseteq \Xi\times  \cS$ from the embedded discrete model 
	(\ref{equation:embedding})/(\ref{equation:expected_utility_discrete_nth_game})/(\ref{equation:rational_pricing_discrete_nth_game})
	that exists by Theorem~\ref{theorem:discrete_trade}. 
	The sequence $(\xi^n_\star)_{n\in\bN}$ of Young measures is obviously tight in the sense of \cite[Definition 3.3]{balder_new_2021} since their support~$[\underline{x},\overline{x}]$ is bounded. From Prohorov's theorem for Young measures \cite[Theorem~4.10]{balder_new_2021}, it follows that $(\xi^n_\star)_{n\in\bN}$ is relatively sequentially compact in the narrow topology. That is, there exists a subsequence $(n_j)_{j\in\bN}$ and a limit point $\xi_\star\in \Xi(E_V, \nu; [\underline{x},\overline{x}])$ such that $\xi^{n_j}_\star \to \xi_\star \quad\text{narrowly as } j \to \infty$. W.l.o.g. $n_j=j$ for all $j\in\bN$.
	
	Next, we apply a version of Koml\'os' theorem to obtain a limiting price function. Since $|S^n|\le 1$ for all $n\in\bN$, \cite[Lemma~A1.1]{delbaen_general_1994} guarantees the existence of a measurable, $E_V$-valued function $S$ and forward convex combinations $(\lambda_{n,k})_{n\in\bN, k=0,1,\ldots,k_n}\subseteq\bR_+$, $k_n\in\bN$ with $\sum_{k=0}^{k_n} \lambda_{n,k}=1$ such that 
	\begin{equation*}
		\sum_{k=0}^{k_n} \lambda_{n,k} S^{n+k}_\star \to S_\star \quad \lambda-\text{a.e. as } n \to \infty.
	\end{equation*}
	It remains to show that $(\xi_\star, S_\star)$ is a Kyle equilibrium in the continuous model.
	
	\smallskip	
	
	{\em Step 2 (Optimality of $\xi_\star$ given $S_\star$):} Let $\xi_0 \in \Xi$. We have to show that
	\begin{equation}\label{equation:optimality}
		u(\xi_0, S_\star) \le u(\xi_\star, S_\star).
	\end{equation}%
	Let us discretise $\xi_0$ according to Lemma~\ref{lemma:approximation_narrow_topology}, such that $\xi_0^n \in \Xi^n$ is a trading strategy of the $n$\textsuperscript{th} discrete game and $\xi_0^n \to \xi_0$ narrowly. 
	But $\xi^n_\star$ from Step~1 is an optimal strategy (in the $n$\textsuperscript{th} discrete game) given $S^n_\star$ and hence $u^n(\xi_0^n, S^n_\star) \leq u^n(\xi^n_\star, S^n_\star)$ for all $n$, and so
	\begin{equation} \label{equation:optimality_approximation}
		\sum_{k=0}^{k_n} \lambda_{n,k} u^{n+k}(\xi_0^{n+k}, S^{n+k}_\star) \leq \sum_{k=0}^{k_n} \lambda_{n,k} u^{n+k}(\xi^{n+k}_\star, S^{n+k}_\star)
	\end{equation}
	for the forward convex combinations~$(\lambda_{n,k})_{n\in\bN, k=0,1,\ldots,k_n}$ from Step~1. Using that $\sum_{k=0}^{k_n} \lambda_{n,k}\allowbreak S^{n+k}_\star \to S_\star$ $\lambda$-a.e., we can apply Lemma~\ref{lemma:convergence_utilities} to both sides of (\ref{equation:optimality_approximation}) to conclude (\ref{equation:optimality}).

	\smallskip
	
	{\em Step 3 (Rationality of $S_\star$ given $\xi_\star$):}
	Our goal is to show that from the rational pricing condition~(\ref{equation:rational_pricing_discrete_nth_game}) of the discrete games for all $n\in\bN$, we can deduce the rational pricing condition~(\ref{equation:rational_pricing_continuous}) of the continuous game.
	Sets of the form
	\begin{equation}\label{equation:pi_system}
		A=[k/2^m,(k+1)/2^m),\ m\in\bN,\ k\in\{(\underline{x}-1)2^m,(\underline{x}-1)2^m+1,\ldots,(\overline{x}+1)2^m-1\}
	\end{equation}
	constitute a $\cap$-stable generator of the $\sigma$-algebra~$\cB(E_Y)$. Thus, by a Dynkin-argument (see, e.g., \cite[Theorem~1.1]{kallenberg_foundations_2021}), it is sufficient to verify (\ref{equation:rational_pricing_continuous}) for $A$ from (\ref{equation:pi_system}).
	We fix a set $A$ of this form with $m\in\bN$. For $n \ge m$, the rational pricing condition~(\ref{equation:rational_pricing_discrete_nth_game}) for the $n$\textsuperscript{th} discrete game implies that
	\begin{equation}\label{equation:rational_pricing_nth_game}
		\begin{split}
			\iiint S^n_\star(x+z) \Ind{A}\left(x+z\right)\, \zeta(dz) \xi^n_\star(v, dx) \nu(dv) \qquad\\\qquad
			=  \iiint \frac{\floor*{2^{n} v}}{2^{n}} \Ind{A}\left(x+z\right)\, \zeta(dz) \xi^n_\star(v, dx) \nu(dv).
		\end{split}
	\end{equation}
	We note that (\ref{equation:rational_pricing_nth_game}) only needs to hold for $n\ge m$ since $\zeta$ is already the limiting measure and 
	integrating the function~$\Ind{A}(x+\cdot\,)$ for a strict subset~$A$ of $[l/2^n,(l+1)/2^n)$ could produce a bias.  
	
	We show the convergence of (\ref{equation:rational_pricing_nth_game}) to (\ref{equation:rational_pricing_continuous}) for the left- and right-hand sides separately, starting with the right-hand side.  
	The function~$(v,x) \mapsto  v \int \Ind{A}\left(x+z\right)\, \zeta(dz)$ is continuous in $x$ by the same reasons which lead to (\ref{equation:continuity_fn}). Since it is furthermore measurable in $v$, it is a suitable integrand for narrow convergence according to \cite[Theorem~2.2(b)]{balder_generalized_1988}. From the narrow convergence of $\xi^n_\star \to \xi_\star$ it follows that 
	\begin{equation*}
		\iiint v\Ind{A}\left(x+z\right)\, \zeta(dz)\xi^n(v,dx) \nu(dv) \to \iiint v\Ind{A}\left(x+z\right)\, \zeta(dz)\xi(v,dx) \nu(dv). 
	\end{equation*}
	Then, the convergence follows from (\ref{equation:integral_approximation_11}) as in the proof of Lemma~\ref{lemma:convergence_utilities}. 
	
	Now we turn to the left-hand side in (\ref{equation:rational_pricing_nth_game}). The convergence of the right-hand side already implies the convergence of the left-hand side (without passing to convex combinations), but it remains to verify that  
	$\iiint S(x+z) \Ind{A}(x+z) \;\zeta(dz)\xi(v, dx)\nu(dv)$ is the limit.
	This follows as in the proof of Lemma~\ref{lemma:convergence_utilities} by passing to convex combinations. The only difference is that instead of (\ref{equation:definition_fn}) we consider the functions~$\wt{f}^{n}(x) := \int  S^{n}(x+z) \Ind{A}(x+z) \,\zeta(dz)$, $n\in\bN$. 
\end{proof}

\begin{remark}
In short, the approach is that averaging over the noise trader's demand~$z_1$ makes the insider's expected gain~$\int x_1(v-S_1(x_1+z_1))\zeta(dz_1)$   
continuous in $(v,x_1)\in E_V\times E_X$, although the price function $S_1$ can be discontinuous. 
Unfortunately, this approach does not work in multi-period models. 
The reason is that, on the one hand, the price~$S_2$ in the second period can discontinuously depend on $x_1+z_1$.
Thus, the gain in the second period,
\begin{equation}\label{17.6.2024.1}
\int x_2(v-S_2(x_1+z_1,x_2+z_2)) \zeta(dz_2),
\end{equation}
is in general not continuous in $(v,x_1,z_1,x_2)$.
On the other hand, the control variable~$x_2$ can be conditioned on $z_1$, and thus averaging the gain~(\ref{17.6.2024.1}) 
over $z_1$ while keeping $x_2$ fixed would not be consistent with the maximisation problem of the insider.
\end{remark}

\appendix
\section{Appendix}

In the proof of Theorem~\ref{theorem:continuous_trade} we need the following lemma, for a condensed proof see, e.g., \cite[Theorem~2.2.8]{bogachev_weak_2018}.

\begin{lemma}\label{lemma:rao_theorem}\cite[Theorem~3.2]{rao_relations_1962}
	Suppose that a sequence of probability measures~$(\mu^n)_{n\in\bN}$ on $(\bR, \cB(\bR))$ converges weakly to a probability measure~$\mu$, and let $F$ be a family of uniformly bounded, pointwise equicontinuous functions from $\bR$ to $\bR$. 
	Then 
	\begin{equation*}
		\sup_{f\in F}\; \abs*{ \int f(y) \mu^n(dy) - \int f(y) \mu(dy)}\to 0\quad\mbox{as}\ n \to \infty.
	\end{equation*}
\end{lemma}

Next we show how an insider strategy in the continuous Kyle game can be approximated in the narrow topology through a sequence of strategies in the discretised games (cf. Subsection~\ref{subsection:results_proofs}).

\begin{lemma}\label{lemma:approximation_narrow_topology}
	Let $\xi \in \Xi(E_V, \nu; E_X)$ be an insider strategy in the continuous Kyle game. There exists a sequence of strategies $\xi^n \in \Xi^n$, $n\in\bN$, in the discrete games such that $(\xi^n)_{n\in\bN}$ converges narrowly to $\xi$. 
\end{lemma}

\begin{proof}
	Let $\xi \in \Xi(E_V, \nu; E_X)$. To shorten formulae, we assume w.l.o.g. that $\nu(\{1\})=\xi(v,\{\overline{x}\})=0$ for all $v\in E_V$. We define approximating Young measures lying in $\Xi^n$ by
	\begin{equation*}
		\begin{aligned}
			\xi^n(v,\,\cdot\,) := \sum_{\substack{k=0\\ \nu(D^n_k)>0}}^{2^n-1} \Ind{D^n_k}(v)\sum_{l=2^n\underline{x}}^{2^n\overline{x}-1}\delta_{l2^{-n}}  
			\frac{\nu \otimes \xi(D^n_k \times D^n_l)}{\nu(D_k^n)} + \sum_{\substack{k=0\\ \nu(D^n_k)=0}}^{2^n-1} \Ind{D^n_k}(v)\, \delta_0,
		\end{aligned}
		\quad
		\begin{aligned}
			v\in E_V,\\ n\in\bN,
		\end{aligned}
	\end{equation*}		
	where $D^n_j:=[j2^{-n},(j+1)2^{-n})$. By \cite[Theorem~2.2(d)]{balder_generalized_1988}, it is sufficient to show
	\begin{equation} \label{equation:narrow_convergence_approximation_lemma}
		\iint \Ind{A}(v) f(x) \,\xi^n(v, dx) \nu(dv) \to \iint \Ind{A}(v) f(x) \,\xi(v, dx) \nu(dv)
	\end{equation}
	for all sets $A=[k/2^m,(k+1)/2^m),\ m\in\bN,\ k\in\{(\underline{x}-1)2^m,(\underline{x}-1)2^m+1,\ldots,(\overline{x}+1)2^m-1\}$ and $f \in \cC_b(E_X)$. For every $n \ge m$, the endpoints of $A$ fall onto the $n$\textsuperscript{th} grid, and we get 
	\begin{align*}
		&\iint \Ind{A}(v) f(x) \xi^n(v, dx) \nu(dv)\\ 
		&\quad= \iint f(x) \sum_{\substack{k=0\\ \nu(D^n_k)>0,\ D^n_k \subseteq A}}^{2^n-1} \sum_{l=2^n\underline{x}}^{2^n\overline{x}-1} \Ind{D^n_k}(v) \frac{\nu \otimes \xi(D^n_k \times D^n_l)}{\nu(D_k^n)} \;\delta_{l 2^{-n}} (dx) \nu(dv)  \\		
		&\quad= \sum_{\substack{k=0\\D^n_k \subseteq A}}^{2^n-1}\sum_{l=2^n\underline{x}}^{2^n\overline{x}-1} f(l2^{-n}) \int_{D^n_k} \xi(v, D^n_l) \nu(dv) \\ 
		&\quad= \iint \Ind{A}(v) f\left( \frac{\floor{2^{n} x}}{2^{n}} \right)  \xi(v, dx) \nu(dv).
	\end{align*}
	By pointwise convergence of $\Ind{A}(v) f(2^{-n}\floor{2^{n} x}) \to \Ind{A}(v)f(x)$ as $n\to\infty$, bounded convergence leads to
	\begin{equation*}
		\iint \Ind{A}(v) f\left(\frac{\floor{2^{n} x}}{2^{n}}\right) \,\xi(v, dx) \nu(dv) \to \iint \Ind{A}(v) f(x) \,\xi(v, dx) \nu(dv),
	\end{equation*}
	which shows that $\xi^n \to \xi$ narrowly.
\end{proof}

\parskip-0.5em\renewcommand{\baselinestretch}{0.9}\small
\setlength{\bibsep}{.5em plus 0.3ex}

\end{document}